\documentclass[12pt]{article}
\usepackage[margin=1in]{geometry}
\usepackage{amsthm,amsmath}
\RequirePackage[numbers]{natbib}
\RequirePackage[colorlinks,citecolor=blue,urlcolor=blue]{hyperref}
\usepackage{amssymb,graphicx, subcaption}
\usepackage{newfloat}
\let\hat\widehat
\let\tilde\widetilde

\newtheorem{theorem}{Theorem}
\newtheorem{lemma}[theorem]{Lemma}
\newtheorem{example}[theorem]{Example}

\newtheorem{proposition}[theorem]{Proposition}

\DeclareMathOperator*{\argmin}{argmin}

\DeclareFloatingEnvironment[
fileext=los,
listname={List of Algorithms},
name=Algorithm,
placement=tbhp,
]{algorithm}

\title{Predictive clustering}

\author{%
	Jaehyeok Shin, Alessandro Rinaldo and Larry Wasserman\\
	Department of Statistics and Data Science \\
	Carnegie Mellon University \\
	\texttt{shinjaehyeok@cmu.edu, arinaldo@cmu.edu, larry@stat.cmu.edu}
}

\begin{document}
	\maketitle
		
		\begin{abstract}
			We show how to convert any clustering into a prediction set. This has the effect of converting the clustering into a (possibly overlapping) union of spheres or ellipsoids. The tuning parameters can be chosen to minimize the size of the prediction set. When applied to $k$-means clustering, this method solves several problems: the method tells us how to choose $k$, how to merge clusters and how to replace the Voronoi partition with more natural shapes. We show that the same reasoning can be applied to other clustering methods.
		\end{abstract}

	\section{Introduction}

	Although $k$-means clustering is very popular,
	it suffers from several drawbacks:
	it is difficult to choose $k$,
	the clusters form a Voronoi tesselation which may not be accurate representations
	of the clusters,
	we may need to merge clusters to get more flexible clusters,
	and the clusters might not approximate the upper level sets of the density.
	
	A simple alternative to $k$-means clustering is to
	use a union of spheres or ellipsoids.
	Specifically, we define the clusters
	to be the connected components
	$C_1,\ldots, C_r$ 
	of the set
	\begin{equation}
	{\cal C} = \bigcup_{j=1}^k B(c_j,r_j)
	\end{equation}
	where
	$B(c_j,r_j)$
	denotes a ball of radius $r_j$ centered at some point $c_j$.
	Note that the number of clusters $r$
	can be strictly less than $k$.
	%We then define the generalized Voronoi diagram
	%$V_1,\ldots, V_r$
	%where $V_j$ is the set of points closer to $C_j$ than any another $C_s$.
	
	Clusters of this form are very simple,
	easy to understand, and easy to represent computationally.
	Of course, we need a way to choose $k$, the centers and the radii.
	One idea would be to choose the spheres to minimize volume subject to
	containing a given fraction of the points.
	This is essentially the excess mass approach pioneered by
	\cite{polonik1995measuring}.
	However, performing this minimization is 
	a combinatorial optimization and is computationally very expensive.
	Also, the excess mass approach does not provide a way to choose $k$.
	
	In this paper, we take a different approach.
	We start with a given clustering method such as $k$-means clustering.
	We define certain residuals from the initial clustering which allows us
	to convert the resulting clusters into a $1-\alpha$ prediction region
	${\cal C}$
	using a method called
	{\em conformal prediction}
	\citep{shafer2005algorithmic}.
	The resulting conformalized clustering turns out to precisely be
	a union of spheres (or ellipsoids).
	Furthermore, we get, as a bonus, the 
	distribution free coverage guarantee, namely,
	\begin{equation}
	\inf_P P^{n+1}(Y\in {\cal C})\geq 1-\alpha
	\end{equation}
	where the infimum is over all probability distributions and
	$Y$ denotes a future observation.
	We then choose any tuning parameters
	(such as $k$ in $k$-means) by minimizing the volume of the prediction region.
	In contrast to using within sums of squares --- which is always minimized by taking $k$ as large as possible ---
	the volume of the prediction set is typically minimized at a value much smaller than the sample size.
	
	In summary, by taking a predictive approach,
	we end up replacing $k$-means clustering with
	$k$-spheres clustering
	and we have a simple way to choose tuning parameters.

	Figure \ref{fig::voronoi}
	shows an example.
	The left panel shows the data.
	The middle panel shows the Voronoi partition from 
	the usual $k$-means algorithm using $k=6$.
	(There are six centers and each element of the Voronoi partition is the set of points closest to each center.)
	The right panel shows the output of our method
	which consists of two clusters (a union of six spheres).
	Despite the fact that we are overfitting,
	the method automatically merges the six spheres
	into two clusters.
	The resulting union of spheres has the aforementioned predictive interpretation.
	A new observation will be contained in the set with probability at least $1-\alpha$
	where, in this example, we chose $\alpha=0.1$.

	We will show that the same general approach can be applied to
	other clustering methods such as mixture models and density-based clustering.
	In each case, by taking a predictive approach
	and minimizing the volume of the prediction set,
	we eliminate the problem of choosing tuning parameters
	and we convert the clusters into
	a simple form such as a union of spheres.

	{\em Related Work.}
	\cite{lei2015conformal}
	used predictive clustering to explore functional data.
	The current paper greatly extends the approach in that paper.
	Clustering based on ellipsoids is considered in
	\cite{kumar2008scale}.
	General excess mass clustering was pioneered
	in \cite{polonik1995measuring}.

	\begin{figure}
		\begin{center}
			\includegraphics[scale=.3]{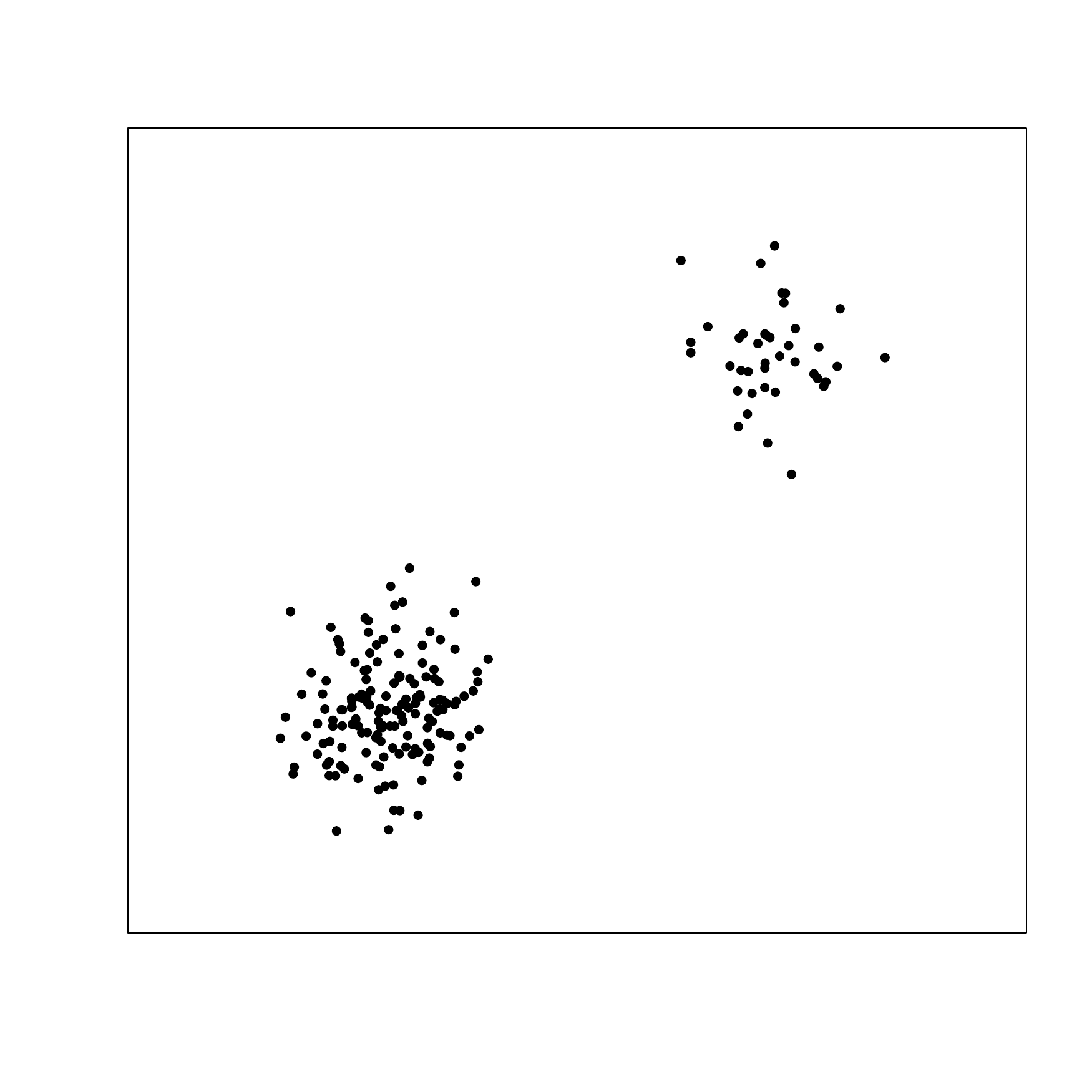}
			\includegraphics[scale=.3]{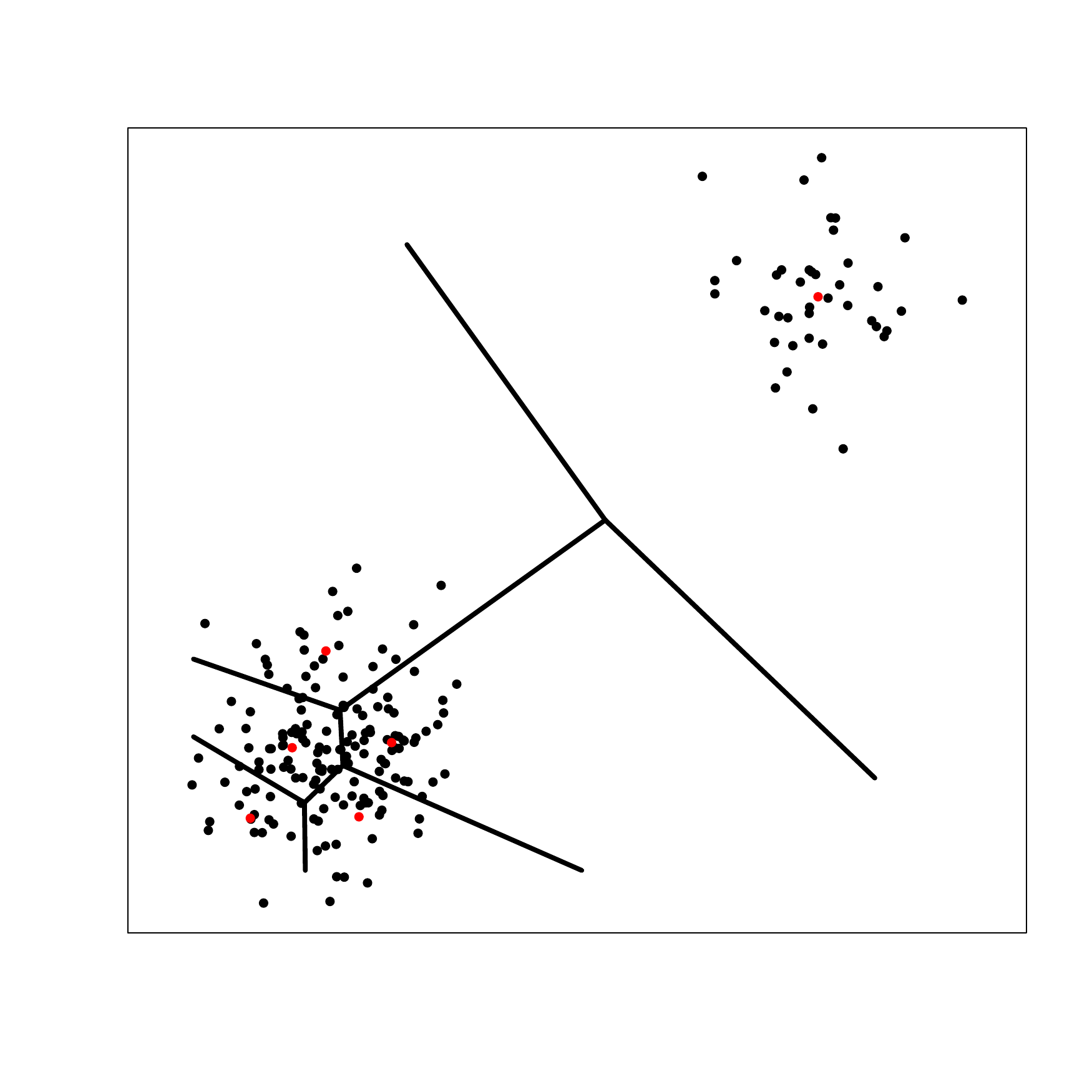}
			\includegraphics[scale=.3]{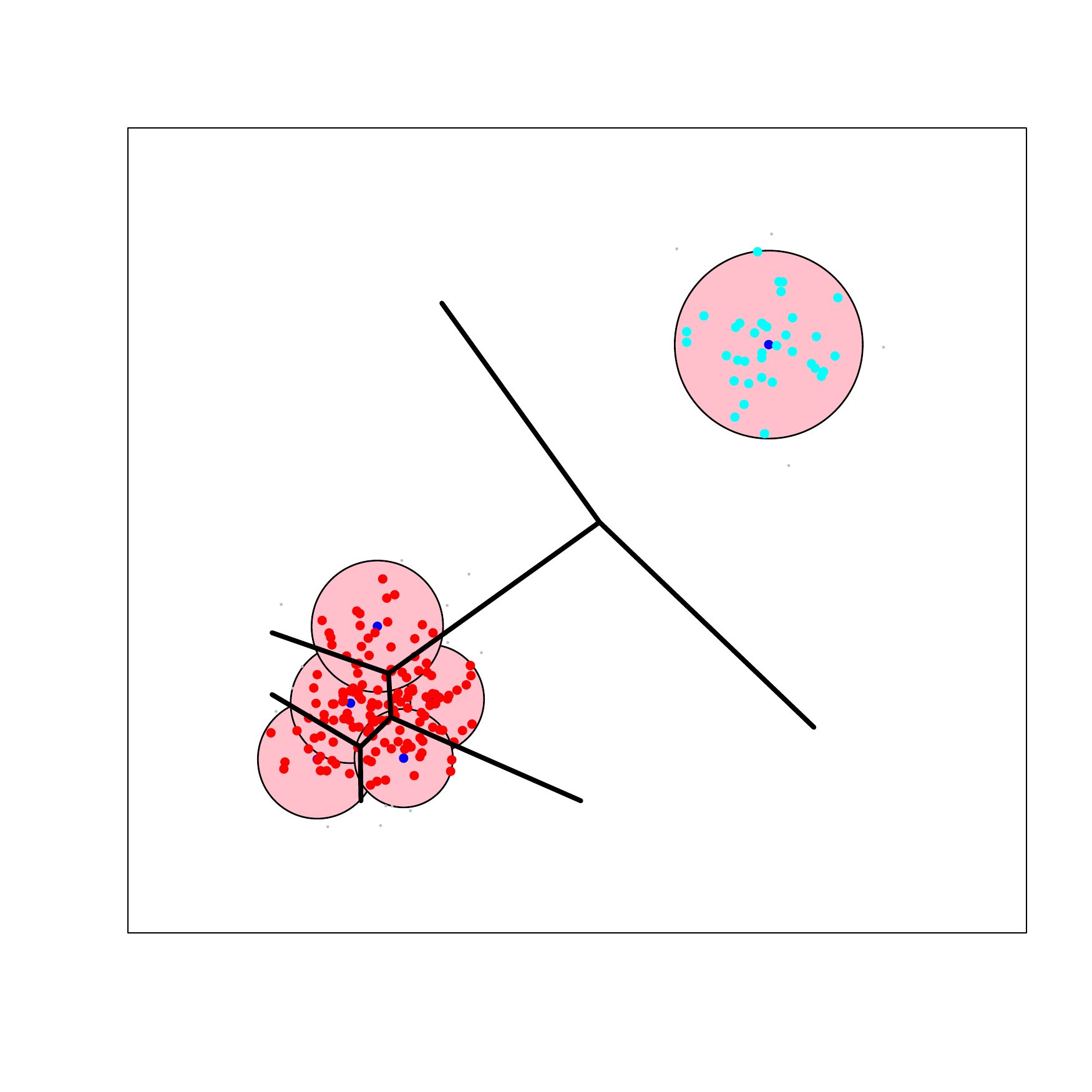}
		\end{center}
		\caption{\em Left: A two-dimensional dataset.
			Middle: Voronoi tesselation from $k$-means clustering with $k=6$.
			Right: Our method converts the 
			Voronoi tesselation into a union of (overlapping) six spheres.
			These spheres merge to form two clusters.}
		\label{fig::voronoi}
	\end{figure}

	{\em Paper Outline.}
	We review $k$-means clustering in Section \ref{section::kmeans}.
	We explain conformal prediction in Section \ref{section::conformal}.
	Our new approach, applied to $k$-means clustering,
	is described in Sections 
	\ref{section::improved} and
	\ref{section::more-improved}.
	Section \ref{section::mixtures}
	shows that we can use the same reasoning to convert
	mixture model clustering into a union of spheres.
	In Section \ref{section::max-mix}
	we introduce a new family of distributions called a 
	max-mixture.
	This allows us to relate unions of spheres with level sets of densities.
	We show that
	level sets of certain functions
	can also be converted into unions of spheres clustering
	in Section \ref{section::level-sets}.

	\section{Background}
	
	\subsection{$k$-means Clustering}
	\label{section::kmeans}
	
	In this section, we briefly review $k$-means clustering, and we discuss some of its drawbacks.
	Let $Y_1,\ldots, Y_n \sim P$ be
	iid draws from a distribution $P$
	where $Y_i\in \mathbb{R}^d$.
	The $k$-means cluster centers
	$c_1,\ldots, c_k\in \mathbb{R}^d$ are chosen to minimize
	\begin{equation}
	R(k)=\frac{1}{n}\sum_{i=1}^n \min_j ||Y_i - c_j||^2.
	\end{equation}
	The Voronoi cell $V_j$ is defined to be
	\begin{equation}
	V_j = \Bigl\{y:\ ||y-c_j|| \leq ||y - c_s|| \ {\rm for \ all\ }s\neq j\Bigr\}.
	\end{equation}
	Thus, $V_j$ consists of the points closest to $c_j$.
	The Voronoi cells $V_1,\ldots, V_k$ define the clusters.
	In particular, the $j^{\rm th}$ data cluster is
	$C_j = \{ Y_i:\ Y_i \in V_j\}$.
	
	Despite its simplicity,
	$k$-means has many drawbacks:
	
	\begin{enumerate}
		\item Despite years of research, there is no agreed upon standard method for choosing $k$.
		\item The Voronoi cells can be rather unnatural shapes that don't reflect what we mean by a cluster.
		\item If the clusters are roughly spherical and well separated, $k$-means works well.
		But in other cases, it can do poorly.
		One way to deal with this problem is to use a large value of $k$ and then
		merge some of the clusters. This gives a much more flexible clustering method.
		But it is not clear how to decide when to merge clusters.
		\item In some sense, the ideal clusters with prediction coverages are the upper level sets of the density $p$ since the level sets are the minimum volume predictive sets with the natural clusters induced by connected components of them. \citep{rinaldo2010generalized, lei2013distribution} 
		One could try to estimate these using density estimation but this can be difficult
		in multivariate problems.
		It would be useful if $k$-means clustering could be used as an approximation to level set clustering.
		But no such results exist.
	\end{enumerate}

	As we shall see, our predictive approach provides solutions, simultaneously,
	to all these problems.
	Before explaining the method,
	we need to discuss conformal prediction.

	\subsection{Conformal Prediction}
	\label{section::conformal}
	
	Conformal prediction 
	\citep{shafer2005algorithmic}
	is a general method for constructing distribution-free
	prediction sets.
	The basic idea is to test the hypothesis 
	$H_0: Y=y$
	that
	a future observation $Y$ takes the value $y$.
	The test is based on a set of residuals or {\em conformal score}.
	The test is performed for all values of $y$.
	By inverting the test, we get a prediction region.
	The precise steps are given in Algorithm~\ref{alg::conformal}.
	
	\begin{algorithm}
		\fbox{\parbox{\textwidth}{
				\begin{center}
					{\sf Conformal Algorithm}
				\end{center}
				\begin{enumerate}
					\item Fix $Y_{n+1} = y$ where $y$ is an arbitrary value.
					Define the {\em augmented dataset} 
					$$
					{\cal A}(y) = \{ Y_1,\ldots, Y_n,Y_{n+1}\}.
					$$
					
					\item 
					For $i=1,\ldots, n+1$ let
					$R_i(y) = \phi(Y_i, {\cal A})$
					where $\phi$ is any function that is
					invariant to permutations of the elements of ${\cal A}$.
					$R_i(y)$ is called the conformal residual.
					\item Let
					\begin{equation}\label{eq::pi}
					\pi(y) = \frac{1}{n+1}\sum_{i=1}^{n+1}I( R_i(y) \geq R_{n+1}(y)).
					\end{equation}
					\item Repeat the above steps for every $y$ and set
					$$
					C_n = \Bigl\{y:\ \pi(y) \geq \alpha\Bigr\} =
					\Biggl\{ y:\ 
					\frac{1}{n+1}\sum_{i=1}^{n+1}
					I\Bigl( R_i(y) \geq R_{n+1}(y)\Bigr) \geq \alpha \Biggr\}.
					$$
				\end{enumerate}
		}}
		\caption{\em The basic conformal algorithm.}
		\label{alg::conformal}
	\end{algorithm}

	The quantity $\pi(y)$ in (\ref{eq::pi}) can be thought of as a p-value
	for testing the hypothesis 
	$H_0: Y=y$.
	The residuals
	$R_1,\ldots, R_{n+1}$ are exchangeable under $H_0$ so that
	$\pi(y)$ is uniformly distributed on
	$\{1/(n+1),\ldots, 1\}$.
	The last step inverts the test to get a confidence set for $Y_{n+1}$.
	It follows that, for every distribution $P$,
	\begin{equation}
	P^{n+1}(Y\in C_n) \geq 1-\alpha
	\end{equation}
	where $Y$ is a new observation.
	For a proof of this fact see \cite{shafer2005algorithmic}.
	
	A simple example of a conformal residual is
	$R_i(y)=|Y_i - \overline{Y}_y|$
	where
	$\overline{Y}_y  = (Y_1 + \cdots + Y_n + y)/(n+1)$ is the mean of the augmented data.
	The resulting prediction set is then
	$$
	C_n = \Biggl\{ y:\ 
	\frac{1}{n+1}\sum_{i=1}^{n+1}
	I\Bigl( |Y_i - \overline{Y}(y)| \geq  |y-\overline{Y}(y)|\Bigr) \geq \alpha \Biggr\}.
	$$

	The choice of residual function $\phi$ does not affect the
	validity --- that is the coverage --- of $C_n$.
	But it does affect the size of $C_n$.
	In some sense, the optimal choice is $R_i = 1/p(Y_i)$;
	see \cite{lei2013distribution}.
	As we explain in the next section, clustering can be used to construct the residuals
	which turns the clustering method into a prediction method.
	A similar idea was used in \cite{lei2015conformal}  for the purpose of
	clustering functional data.
	
	{\em Split Conformal Prediction.}
	A modified algorithm called
	{\em split conformal prediction}
	is much faster than the standard conformal method.
	As the name implies, the method uses data splitting.
	Suppose the residuals have the form
	$R_i = \phi(Q,Y_i)$
	for some function $Q$ of the data.
	We use half the data to compute $Q$, and we compute the residuals on the other half of the data.
	We then get a prediction region by choosing the appropriate quantile of the residuals.
	The details are in Algorithm~\ref{alg::split}.

	\begin{algorithm}
		\fbox{\parbox{\textwidth}{
				\begin{center}
					{\sf Split Conformal Algorithm}
				\end{center}
				\begin{enumerate}
					\item Split the data into two halves
					${\cal Y}_1$ and 
					${\cal Y}_2$.
					\item Compute $Q$ from ${\cal Y}_1$.
					\item Compute the residuals
					$R_i = \phi(Q,Y_i)$ where
					$Y_i \in {\cal Y}_2$.
					\item Let $t_\alpha$ be the $1-\alpha$ quantile of the residuals.
					\item Let $C_n= \{ y:\ \phi(Q,y) \leq t_\alpha\}.$
				\end{enumerate}
		}}
		\caption{\em The split conformal algorithm.}
		\label{alg::split}
	\end{algorithm}
	
	The set $C_n$ again has the distribution-free property
	$$
	\inf_P P(Y_{n+1}\in C_n) \geq 1-\alpha.
	$$
	Split conformal inference is much faster than full conformal inference
	because it avoids the augmentation step.

	\section{From $k$-Mean Clustering to $k$-Spheres Clustering}
	\label{section::improved}
	
	We can combine prediction with $k$-means clustering 
	by defining residuals from the clustering and then applying
	conformal prediction.
	As we shall see,
	the resulting prediction set can be used to define modified clusters.
	We will use split conformal prediction.
	So the first step is to split randomly
	the data into two groups
	${\cal Y}_1$ and ${\cal Y}_2$.
	For simplicity, assume each group has a size of $n$.
	
	\subsection{The Basic Method}

	Let
	$c_1,\ldots, c_k$ denote the cluster centers
	after applying $k$-means clustering to ${\cal Y}_1$.
	Define the residual
	$$
	R_i = \min_j ||Y_i - c_{j}|| :=  ||Y_i - c_{j(i)}||
	$$
	for each $Y_i \in {\cal Y}_2$.
	Here, $c_{j(i)}$ is the cluster center closest to $Y_i$.
	Now apply the conformal algorithm from the last section
	to get a conformal prediction set $C_k$.
	The precise steps are given in Algorithm~\ref{alg::k-spheresI}.
	
	\begin{algorithm}
		\fbox{\parbox{\textwidth}{
				\begin{center}
					{\sf Converting $k$-means to $k$-Spheres}
				\end{center}
				\begin{enumerate}
					\item Split the data into two halves
					${\cal Y}_1$ and 
					${\cal Y}_2$.
					\item Run $k$-means on ${\cal Y}_1$ to get centers
					$c_1,\ldots, c_k$.
					\item For the data in ${\cal Y}_2$ compute the (non-augmented) residuals
					$$
					R_i = ||Y_i - c_{j(i)}||
					$$
					where $c_{j(i)}$ is the closest center to $Y_i$.
					\item Let $t_\alpha$ be the $1-\alpha$ quantile of the residuals.
					\item Let ${\cal C}_k= \{ y:\ \min_j ||y - c_j|| \leq t_\alpha\}$.
					\item Choose $\hat k$ to minimize $\mu({\cal C}_k)$ where $\mu$ is Lebesgue measure.
					\item Return: ${\cal C}_{\hat k}$.
				\end{enumerate}
		}}
		\caption{\em This algorithm converts $k$-means clustering into $k$-spheres clustering.}
		\label{alg::k-spheresI}
	\end{algorithm}
	
	From the definition of the residuals,
	together
	with the definition of ${\cal C}_k$ we immediately see
	that ${\cal C}_k$ is a union of spheres.
	We this have:
	
	\begin{lemma}
		The conformal set has the form
		\begin{equation}
		{\cal C}_k = \bigcup_{j=1}^k B(c_j, t_\alpha)
		\end{equation}
		where $t_\alpha$ is defined in Step 4 of Algorithm~\ref{alg::k-spheresI}.
		If $Y$ denotes a new observation, then
		\begin{equation}
		\inf_P P^{n+1}(Y \in {\cal C}_k)\geq 1-\alpha
		\end{equation}
		where
		the infimum is over all distributions $P$.
	\end{lemma}

	Every choice of $k$ gives a valid prediction set ${\cal C}_k$
	with the correct coverage.
	We choose
	$k$ to minimize the Lebesgue measure of ${\cal C}_k$.
	Unlike the within sums of squares --- which strictly decreases with $k$ ---
	the Lebesgue measure typically decreases then, eventually increases as $k$ increases.
	Thus we define
	\begin{equation}
	\hat k = \argmin_k \mu({\cal C}_k).
	\end{equation}
	The final clustering is ${\cal C}_{\hat k}$.
	We have now replaced the Voronoi tesselation with
	more intuitively shaped sets.
	Also, some of the spheres may be connected (overlapping).
	In other words, we have automatically merged the clusters.

	{\em Remark 1:
		To check whether spheres are connected to each other, we examine whether there exists at least a single point in intersections of spheres. This sample-based connectivity checking rule might possibly fail to find some weak connections among spheres.  In practice, however, we can obtain stable clusters by using the sample-based rule to disconnect weakly connected spheres.
	}
	
	{\em Remark 2:
		To preserve the prediction guarantee,
		we can restrict $k$ to be in some large range
		$1\leq k\leq K_n$
		and replace $\alpha$ with $\alpha/K_n$.
		Here, $K_n$ can increase with $n$.
		The union bound then ensures that the coverage guarantee is still valid
		even after selecting $k$.
		However, as we view clustering mainly as an exploratory tool,
		it is not crucial to do this correction.}
	
	The Lebesgue measure $\mu(C_k)$ can be easily estimated
	using importance sampling.
	For example, we can draw
	$Z_1,\ldots, Z_N$ from any convenient density $g$
	(such as a uniform density over a large region)
	and then we have that
	$$
	\mu(C_k) = \int_{C_k}d\mu = \int_{C_k} \frac{g(u)}{g(u)}d\mu(u) \approx 
	\frac{1}{N}\sum_{j=1}^N \frac{I(Z_j\in C_k)}{g(Z_j)}.
	$$

	The residual
	$R_i= ||Y_i - c_{j(i)}||$
	is simple and intuitive.
	But, due to the generality of conformal prediction,
	we have the liberty of choosing any permutation invariant residual.
	In fact, much better clusters can be obtained by defining
	\begin{equation}\label{eq::resid}
	R_i = \min_{j \in [k]} \left[\frac{ ||Y_i - c_{j}||^2}{\hat\sigma_{j}^2} + 
	2 d \log \hat\sigma_{j} - 2 \log \hat\pi_{j}\right]
	\end{equation}
	where
	$\hat \pi_n = n_j/n$,
	$n_j$ is the number of points from ${\cal Y}_1$ in $V_j$,
	and
	$$
	\hat\sigma_j^2 = n_j^{-1}\sum_{Y_i\in C_j} ||Y_i - \overline{Y}_j||^2.
	$$
	Let $t_\alpha$ be the $1-\alpha$ quantile of the residuals.
	Now we use Algorithm~\ref{alg::k-spheresI}
	but with the residuals defined as in
	(\ref{eq::resid}).
	Denote the resulting conformal prediction region by
	${\cal M}_k$.
	From the definition of $M_k$ together with the basic property of
	the conformal prediction we have:
	
	\begin{lemma}
		We have that
		$$
		{\cal M}_k = \bigcup_{j=1}^k B(c_j,r_j)
		$$
		where
		$$
		r_j = \hat\sigma_j \sqrt{[t_\alpha + 2\log \hat\pi_j - 2d\log \hat\sigma_j]_+}.
		$$
		Furthermore,
		$\inf_P P^{n+1}(Y \in {\cal M}_k)\ge 1-\alpha$.
	\end{lemma}

	As before we choose $k$ to minimize
	$\mu({\cal M}_k)$.

	\begin{example}
		Figure \ref{fig::kmeans} illustrate how the algorithm for converting
		$k$-means clustering into $k$-spheres clustering works. The left
		panel shows the data generated from four Normal distributions with
		background noise. The middle panel shows how the volume of clusters
		changes as $k$ varying from $1$ to $20$. The right panel shows
		clusters with the minimum volume ($k = 4$). The resulting clusters are
		$1-\alpha$ predictive region with $\alpha = 0.1$.
	\end{example}
	
	\begin{figure}
		\begin{center}
			\includegraphics[scale=.33]{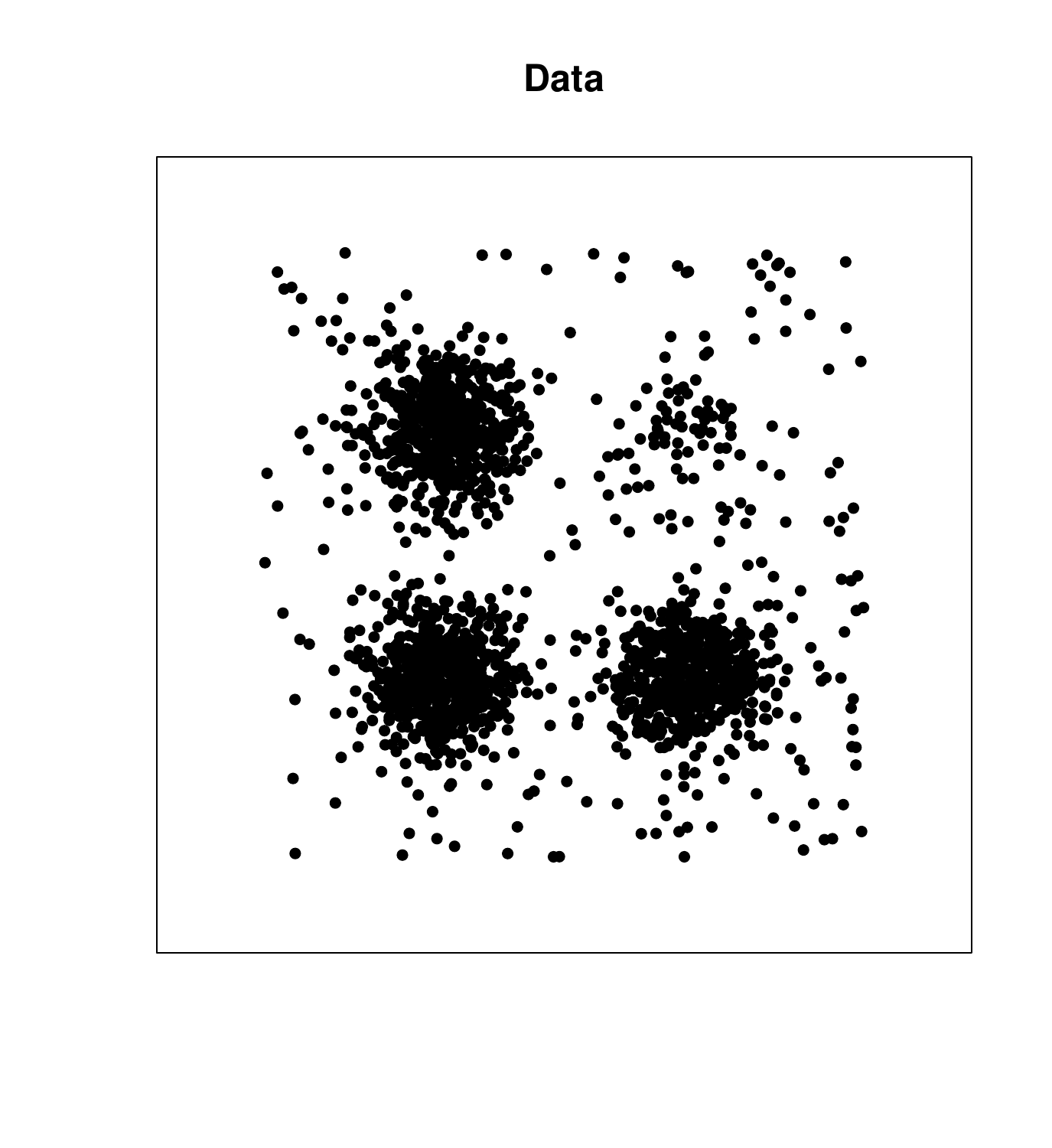}
			\includegraphics[scale=.33]{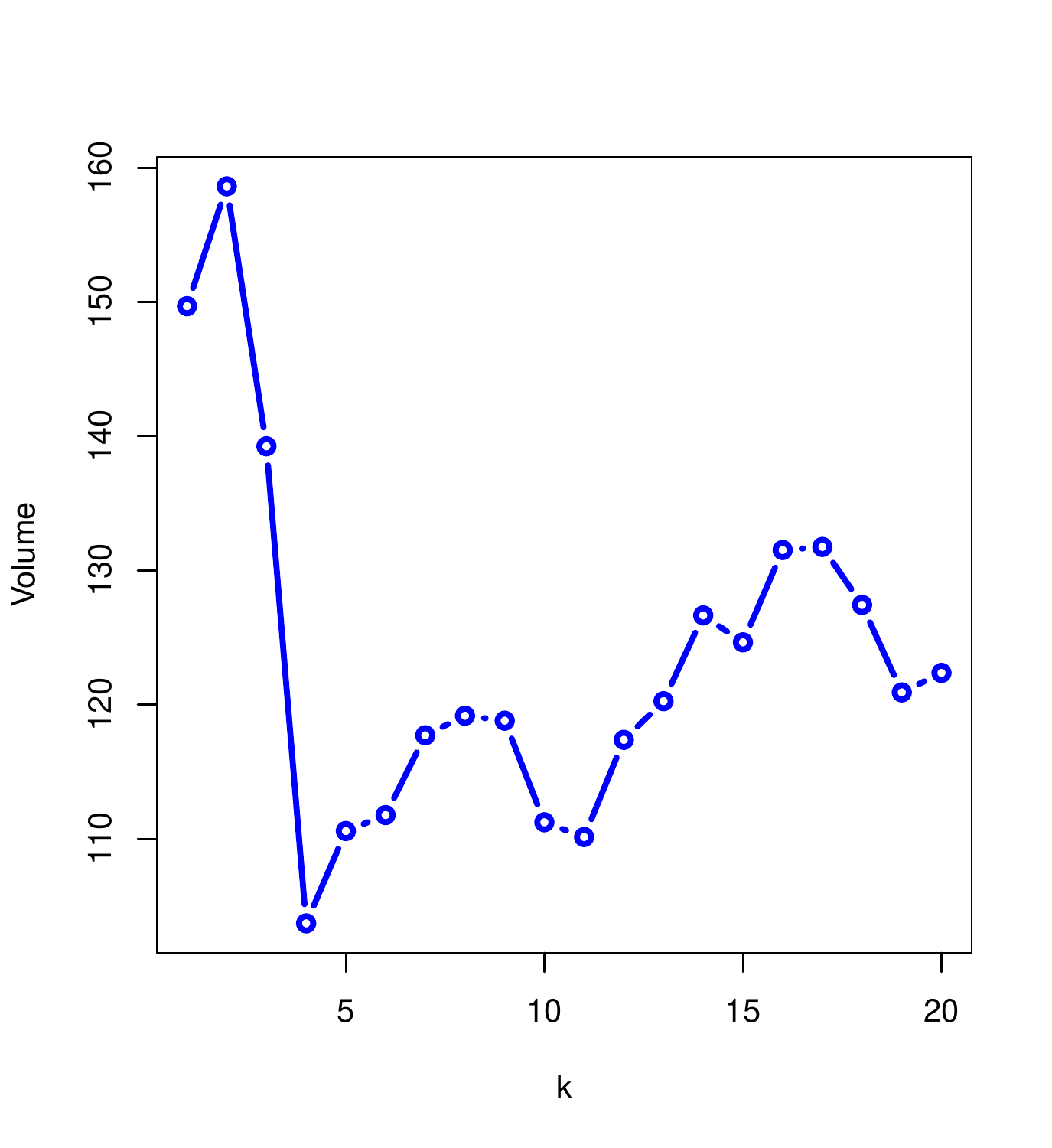}
			\includegraphics[scale=.33]{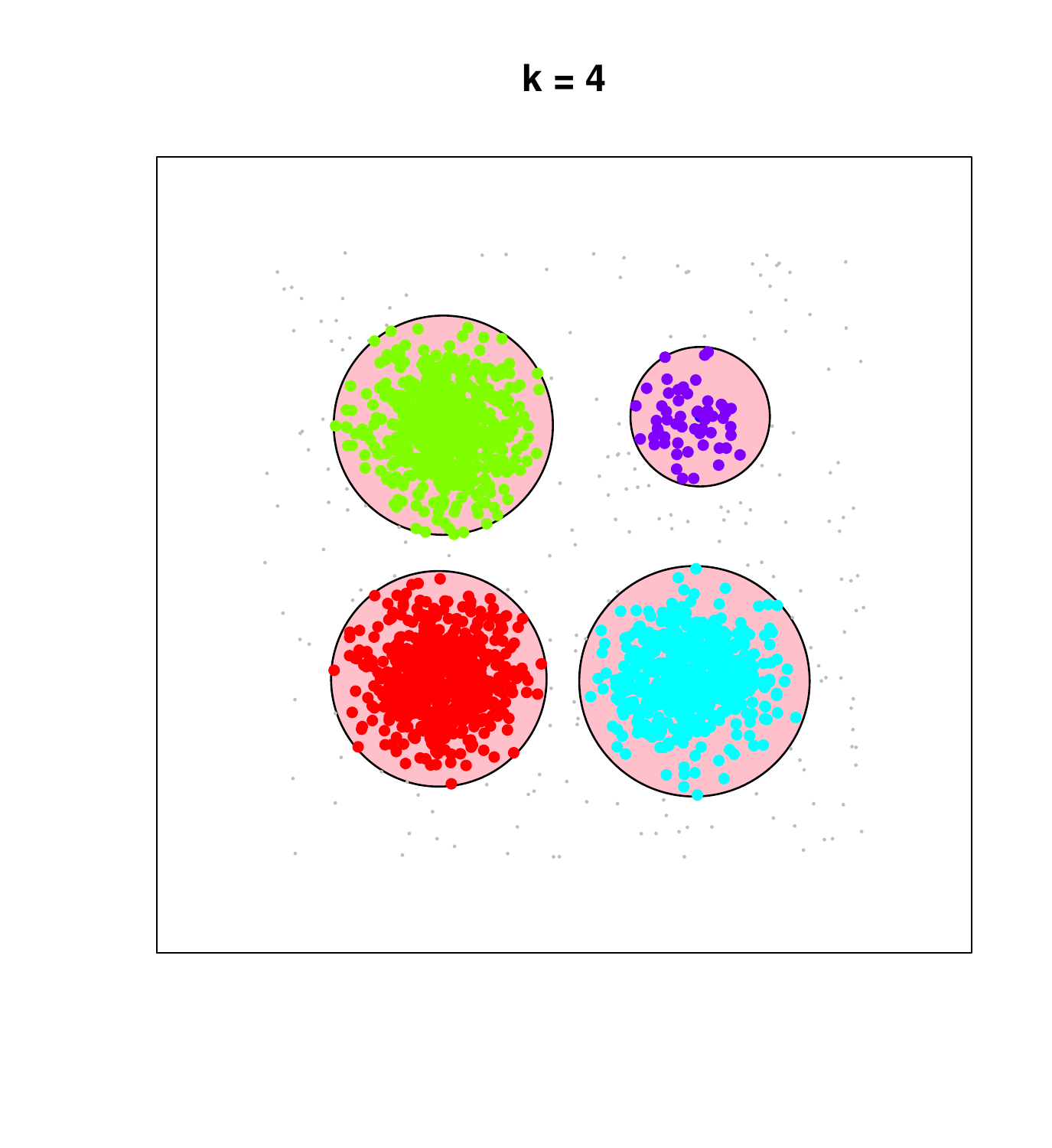}
		\end{center}
		\caption{\em Left: A two-dimensional dataset.
			Middle: $k$ versus volume of clusters.
			Right: Clusters with the minimum volume.}
		\label{fig::kmeans}
	\end{figure}

	Note that if we chose $k$ to minimize the usual within sums-of-squares,
	then we always end up choosing $k=n$.
	This is why choosing $k$ is so hard in $k$-means clustering.
	Fortunately, this does not happen in our approach.
	In fact, if the data points are well-separated into $k^* < n$ spherical clusters, 
	then $\mu(M_k)$ is minimized at $k = k^* < n$.
	We see this in Figure \ref{fig::kmeans_examples}, which shows what happens when we let $k$ grow.
	
	\begin{figure}
	    \begin{center}
	        \includegraphics[scale=.27]{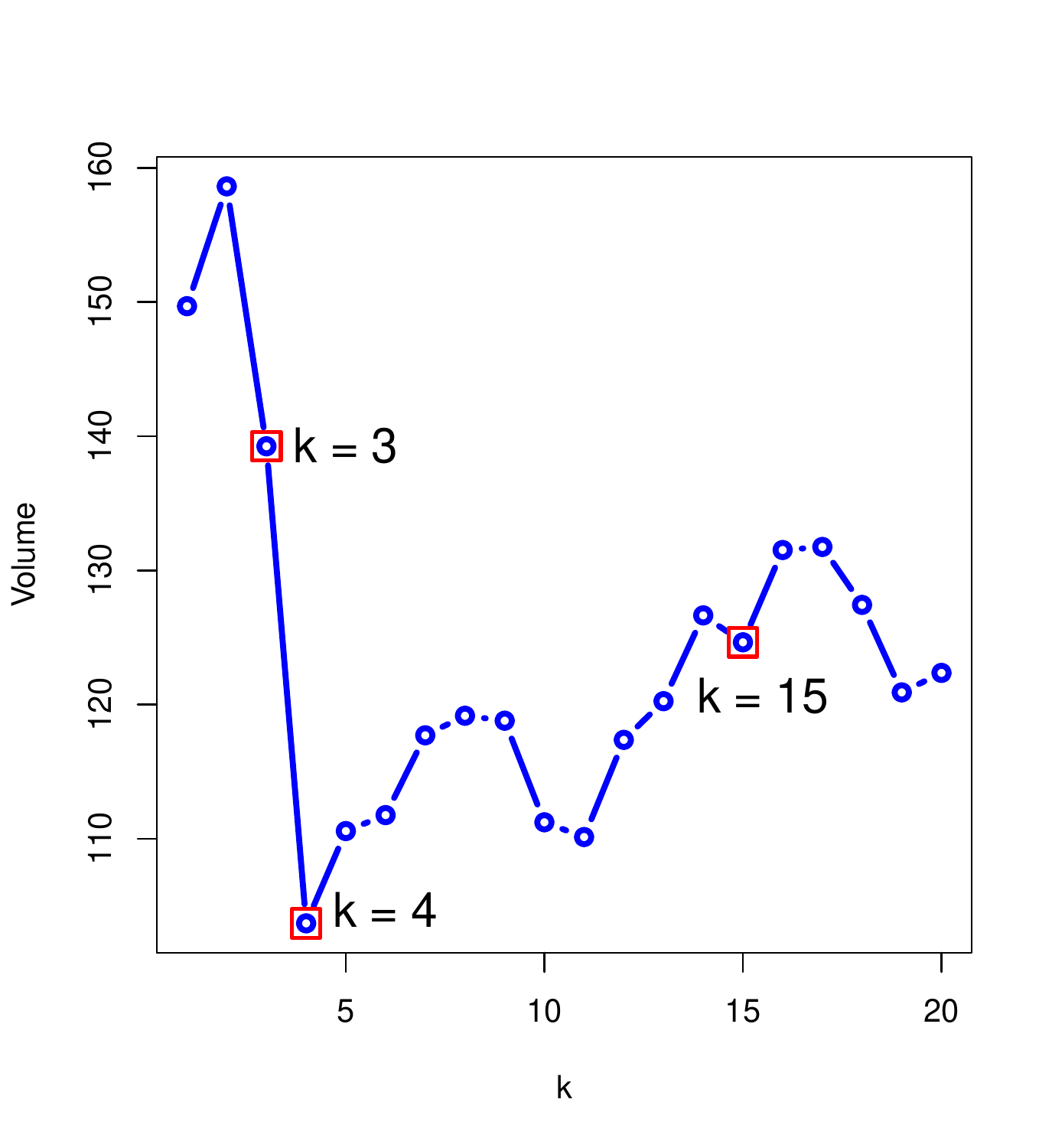}
	        \includegraphics[scale=.27]{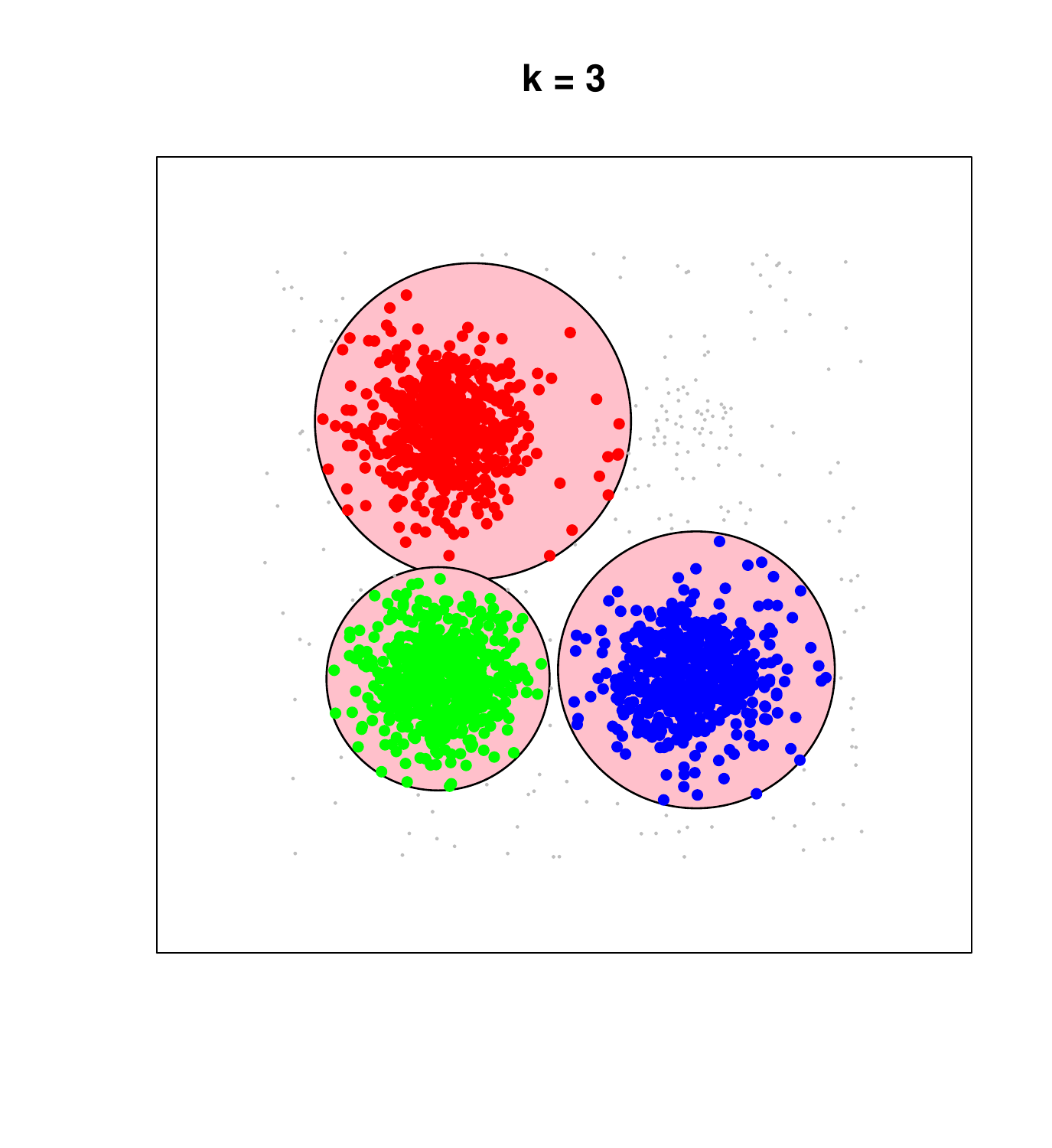}
	        \includegraphics[scale=.27]{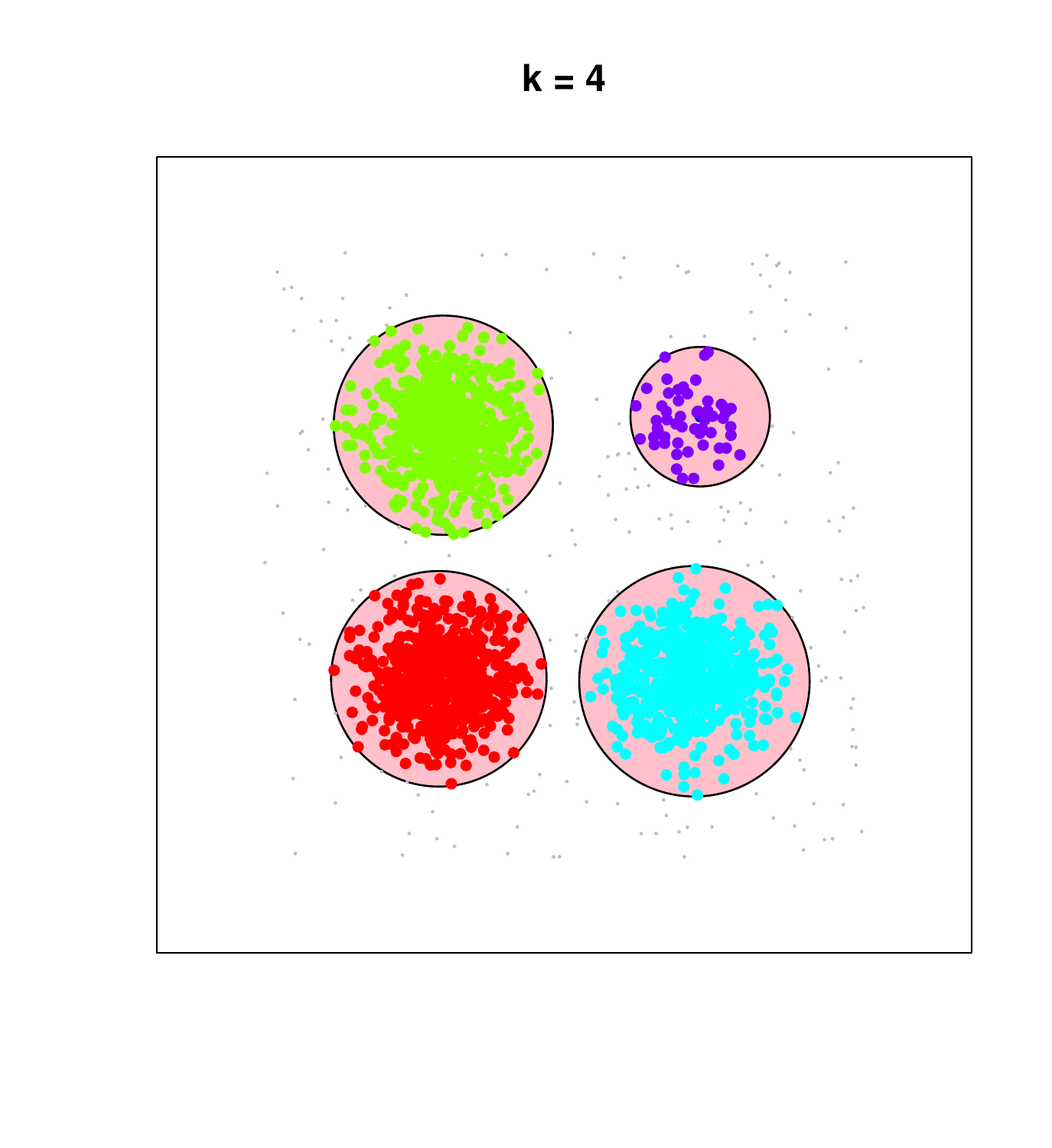}
	        \includegraphics[scale=.27]{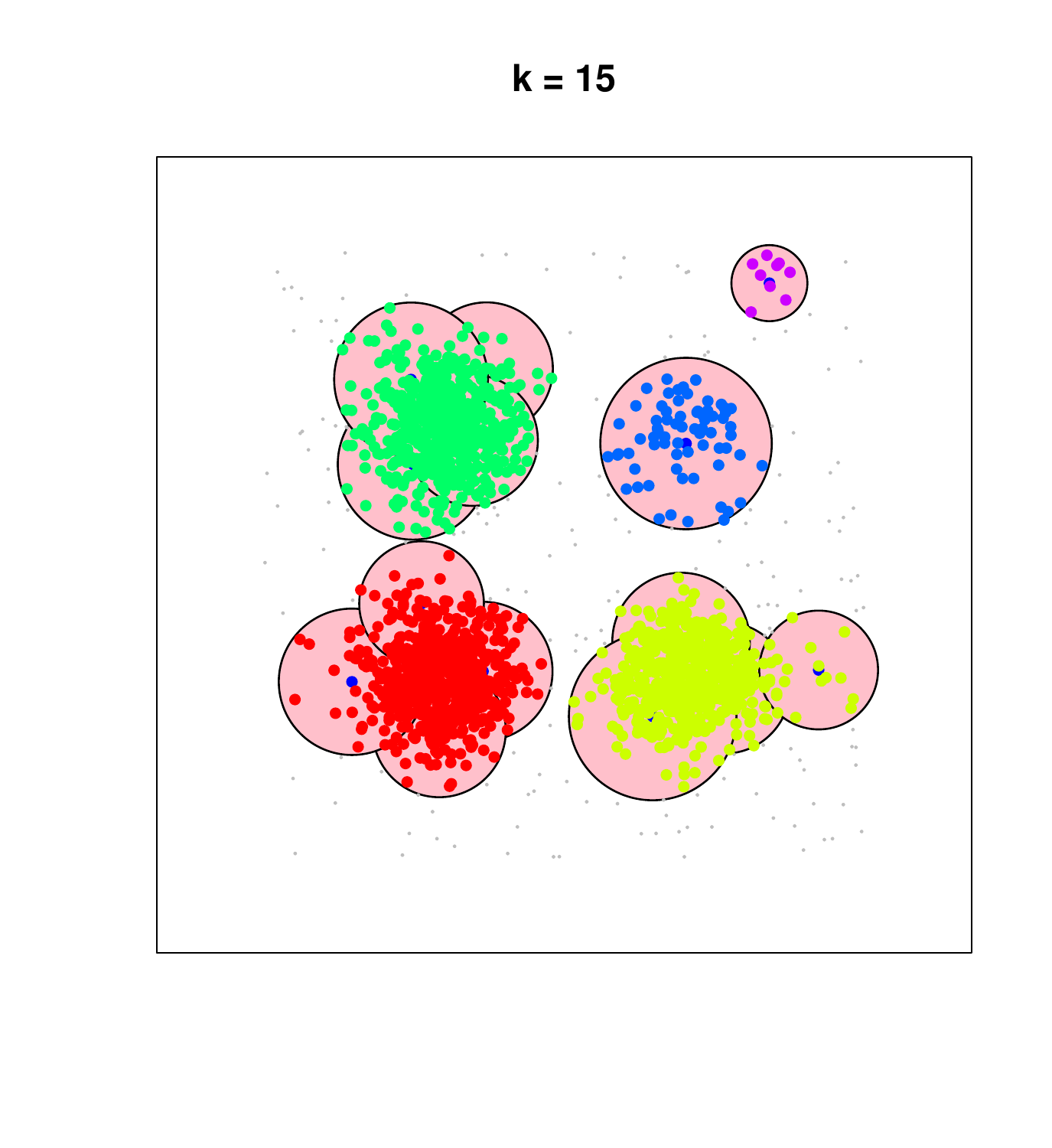}
	    \end{center}
	    \caption{\em From left to right : Volumes of clusters and clusters with $k = 3, 4, 15$.}
	    \label{fig::kmeans_examples}
	\end{figure}
%	
%	\begin{figure}
%		\centering
%		\begin{subfigure}[b]{0.45\textwidth}
%			\centering
%			\includegraphics[width=\textwidth]{Figures/Kmeans_vol_curve_with_dots}
%		\end{subfigure}
%		\hfill
%		\begin{subfigure}[b]{0.45\textwidth}  
%			\centering 
%			\includegraphics[width=\textwidth]{Figures/Kmeans_k_3}
%		\end{subfigure}
%		\vskip\baselineskip
%		\begin{subfigure}[b]{0.45\textwidth}   
%			\centering 
%			\includegraphics[width=\textwidth]{Figures/Kmeans_k_4}
%		\end{subfigure}
%		\quad
%		\begin{subfigure}[b]{0.45\textwidth}   
%			\centering 
%			\includegraphics[width=\textwidth]{Figures/Kmeans_k_15}
%		\end{subfigure}
%		\caption{\em (a) $ k$ versus volume of clusters.
%			(b-d) clusters with $k = 3, 4, 15$.}
%		\label{fig::kmeans_examples}
%	\end{figure}
%	

	\subsection{Improving the Choice of $k$ With Hypothesis Testing}
	\label{section::more-improved}
	
	So far we have chosen $\hat k$ to minimize $S_k \equiv \mu(C_k)$.
	Now, $S_k$ is a random variable
	and we can end up choosing a large $k$ due to random fluctuations.
	The resulting clustering, even in such cases, can be quite reasonable.
	Still, it is sometimes useful to use simpler clusterings with small $k$ when possible.
	To this end,
	we can use a hypothesis testing approach to choose $k$.

	%%\textcolor{red}{
	%%Let $\nu_k = \mathbb{E}[S_k]$ be the mean of $S_k$.
	%%For $k=1,2,\ldots, $ we will test the hypothesis
	%%$$
	%%H_0: \nu_k \leq \nu_t\ \ {\rm for\ all\  t>k}\ \ \ \ \ {\rm versus}\ \ \ \ \ 
	%%H_1: \nu_t < \nu_k\ \ {\rm for\ some\  t>k}.
	%%$$
	%%We choose $\hat k$ to be the first $k$ for which this hypothesis is not rejected.
	%%    We will construct a confidence interval $A_k$
	%%for the parameter
	%%\begin{equation}
	%%\theta_k = \min_{t>k}(\nu_t - \nu_k).
	%%\end{equation}
	%%If every value in $A_k$ is positive, we reject $H_0$.
	%%    To construct $A_k$
	%%we use the bootstrap.
	%%Let
	%%$\nu = (\nu_1,\ldots, \nu_K)$
	%%where $K$ is the largest value of $k$ we are considering.
	%%Let
	%%$S = (S_1,\ldots, S_K)$. We approximate the cdf
	%%$$
	%%F_n(t) = P( \sqrt{n}||S-\nu||_\infty \leq t)
	%%$$
	%%with its bootstrap approximation
	%%$$
	%%\hat F_n(t) = \frac{1}{B}\sum_{j=1}^B I( \sqrt{n}||S^{(j)}-S||_\infty \leq t)
	%%$$
	%%where
	%%$S^{(1)},\ldots, S^{(B)}$ denote bootstrap replications.
	%%Let $z_\alpha = \hat F_n^{-1}(1-\alpha)$.
	%%The set
	%%$W = \{ \nu:\ ||\nu-S||_\infty \leq z_\alpha/\sqrt{n}\}$
	%%is an asymptotic $1-\alpha$ confidence set for $\nu$.
	%%And
	%%$A_k = \{ \min_{s>k}(\nu_k - \nu_s):\ \nu_k,\nu_s \in W\}$
	%%is an asymptotic $1-\alpha$ confidence set for 
	%%$\min_{s>k} (\nu_k - \nu_s)$.
	%%We reject $H_0$ if
	%%$\min\{ a\in A_k\} >0$.
	%%}
	%%

	Let $\nu_k = \mathbb{E}[S_k]$ be the mean of $S_k$.
	For $k=K, K-1,\ldots, 1$ we will test the hypothesis
	$$
	H_0^k: \nu_t < \nu_k\ \ {\rm for\ some\  t < k}.\ \ \ \ \ {\rm versus}\ \ \ \ \ 
	H_1^k: \nu_k \leq \nu_t\ \ {\rm for\ all\  t < k}.
	$$
	We choose $\hat k$ to be the first $k$ for which this hypothesis is rejected.
	We will construct a confidence interval $A_{t,k}$
	for the parameter
	\begin{equation}
	\theta_{t, k} = \nu_t - \nu_k.
	\end{equation}
	For each $k$, if every value in $\bigcup_{t <k} A_{t,k}$ is positive, we reject $H_0^k$.
	To construct $A_{t,k}$
	we use the bootstrap. For each pair of $t < k$, we approximate the cdf
	$$
	F_n^{k} (t) = P( \sqrt{n}|S_t - S_k-\theta_{t, k}| \leq t)
	$$
	with its bootstrap approximation
	$$
	\hat F_n^{k}(t) = \frac{1}{B}\sum_{j=1}^B I( \sqrt{n}|S_t^{(j)}-S_k^{(j)} -(S_t - S_k)| \leq t)
	$$
	where
	$S^{(1)},\ldots, S^{(B)}$ denote bootstrap replications.
	Let
	\begin{align*}
	&z_{1-{\alpha/2}}= \hat F_n^{-1}(1-\alpha/2) \\
	&z_{\alpha/2} = \hat F_n^{-1}(\alpha/2) 
	\end{align*}
	and let
	\[
	A_{t,k} = \left\{\theta : 
	2(S_t - S_k) - z_{1-\alpha/2/(k-1)} < \theta < 2(S_t - S_k) - z_{\alpha/2/(k-1)}\right\},~~\forall t = 1, \dots, k-1
	\]
	Then, $A_{1,k}, \dots, A_{k-1, k}$ are asymptotic $1-\alpha$ confidence sets for $\theta_{1, k}, \dots, \theta_{t, k}$.
	We reject $H_0$ if
	$\min\{ a\in \bigcup_{t <k} A_{t,k}\} >0$.

	\section{Mixture Models}
	\label{section::mixtures}
	
	The same approach can be used with mixture models.
	Let
	$$
	p_k(y) =\sum_{j=1}^k \pi_j \phi(y;\mu_j,\Sigma_j)
	$$
	denote a mixture of Gaussians
	where
	$\phi(x;\mu_j,\Sigma_j)$ denotes a Normal
	density with mean vector $\mu_j$ and covariance matrix $\Sigma_j$.
	In this case, we define the residual to be
	\begin{equation} \label{eq::resid_GMM}
	R_i = \min_{j \in [k]} \left[\frac{1}{\hat \pi_{j} \phi(Y_i;\hat\mu_{j},\hat\Sigma_{j})}\right]
	\end{equation}
	or
	\begin{equation} \label{eq::resid_GMM2}
	R_i =     \min_j \left[ \frac{1}{2} (Y_i-\mu_j)^T \Sigma_j^{-1} (Y_i -\mu_j) + 
	\frac{1}{2}\log \det(\Sigma_j) - \log \pi_j \right].
	\end{equation}
	
	We note that a similar residual was used in
	\cite{lei2015conformal}
	in the context of functional clustering.
	The resulting conformal prediction set is a union of ellipses:
	\begin{equation}
	{\cal C}_k =\bigcup_{j=1}^k E(\hat\mu_j,\hat\Sigma_j,r_j)
	\end{equation}
	where
	$t_\alpha$ is the $1-\alpha$ quantile of the residuals,
	$E(a,A,b) = \{y:\ (y-a)^T A^{-1}(y-a)\leq b^2\}$ and
	$$
	r_j^2 = \left[\log \hat \pi_j + \log t_\alpha - \frac{d}{2}\log (2\pi) - \frac{1}{2}\log |\hat\Sigma_j|\right]_+.
	$$
	The clusters are the connected components
	of ${\cal C}_k$.
	Again, we see that
	the method automatically merges components that are near each other
	and $k$ can be chosen to minimizing the volume of
	${\cal C}_k$.
	
	Note that our residual
	is precisely the residual in
	(\ref{eq::resid})
	except that $k$-means replaces
	maximum likelihood
	and Euclidean distance replaces
	Mahalanobis distance.
	Thus the method based on
	(\ref{eq::resid})
	can be thought of as an approximate
	approach to conformalized mixtures of Gaussians.
	This connection can be made more formal, as we show in the next section.

	\begin{example}
		Figure \ref{fig::GMM} illustrates how the algorithm for converting
		mixture models into $k$-ellipsoids clustering works. The left
		panel shows the data generated from three Normal distributions
		with different covariance structures in addition to background
		noises. The middle panel shows how the volume of clusters
		changes as $k$ varying from $1$ to $10$. The shaded lines
		corresponds to volume curves with bootstrapped data.  The
		right panel shows clusters with $k = 4$ which is selected by
		the testing based method in the previous section. Note that
		there are only three ellipsoids even if the clustering is based on
		mixture of four Normal distributions since $r_j = 0$ for one of
		ellipsoidal clusters (blue dot in the middle). The resulting
		clusters are $1-\alpha$ predictive region with $\alpha = 0.1$.
	\end{example}

	\begin{figure}
		\begin{center}
			\includegraphics[scale=.33]{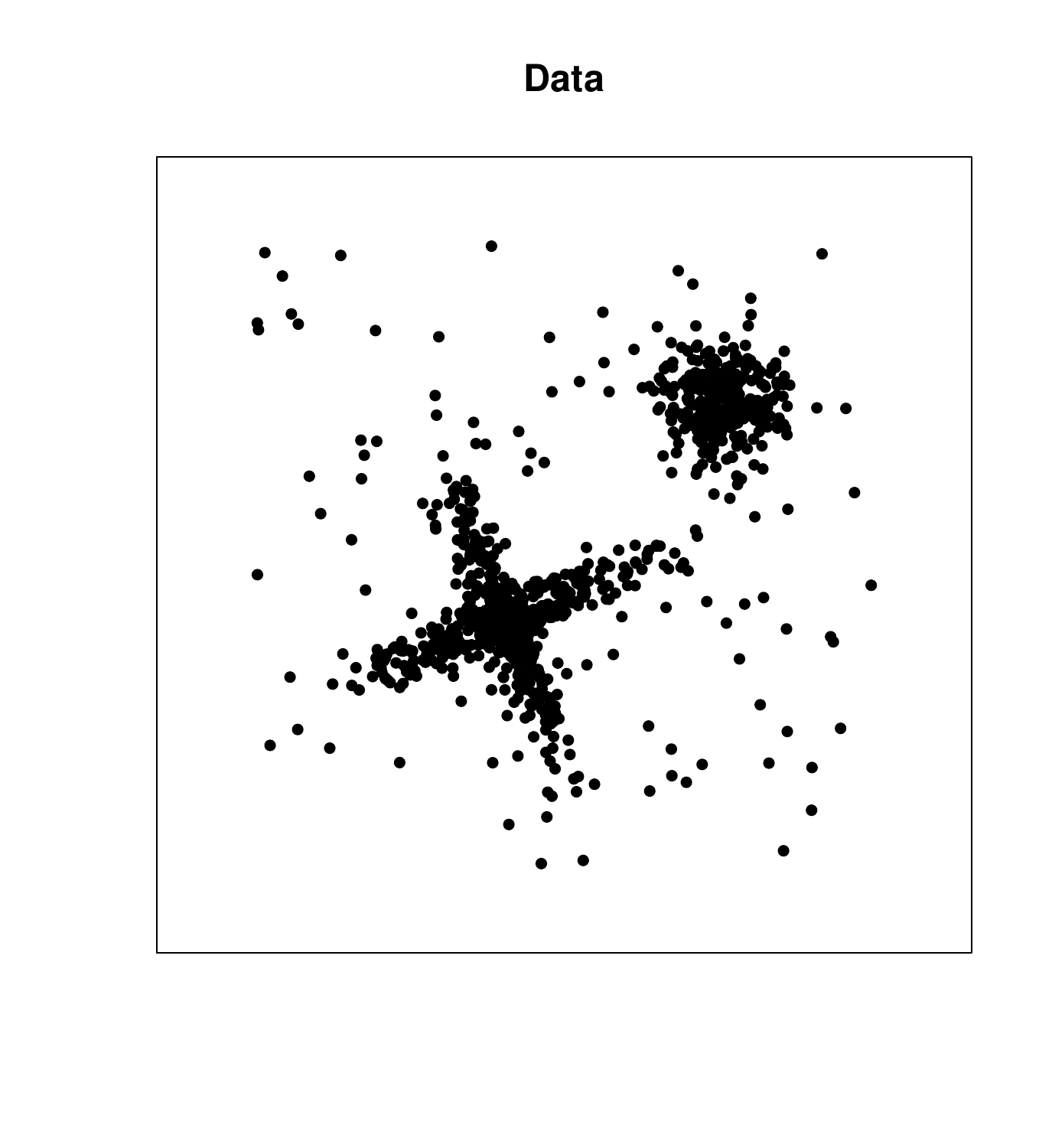}
			\includegraphics[scale=.33]{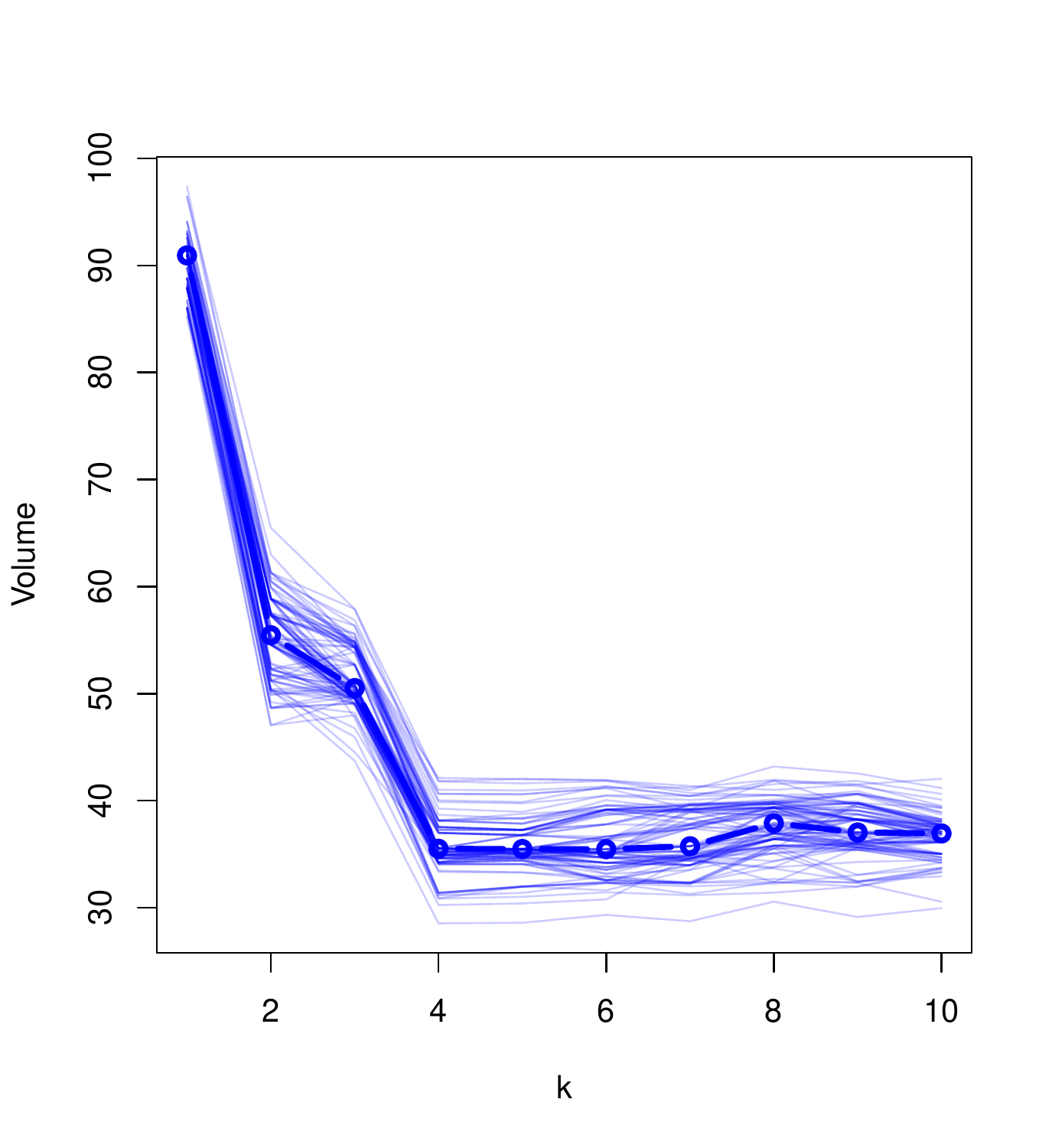}
			\includegraphics[scale=.33]{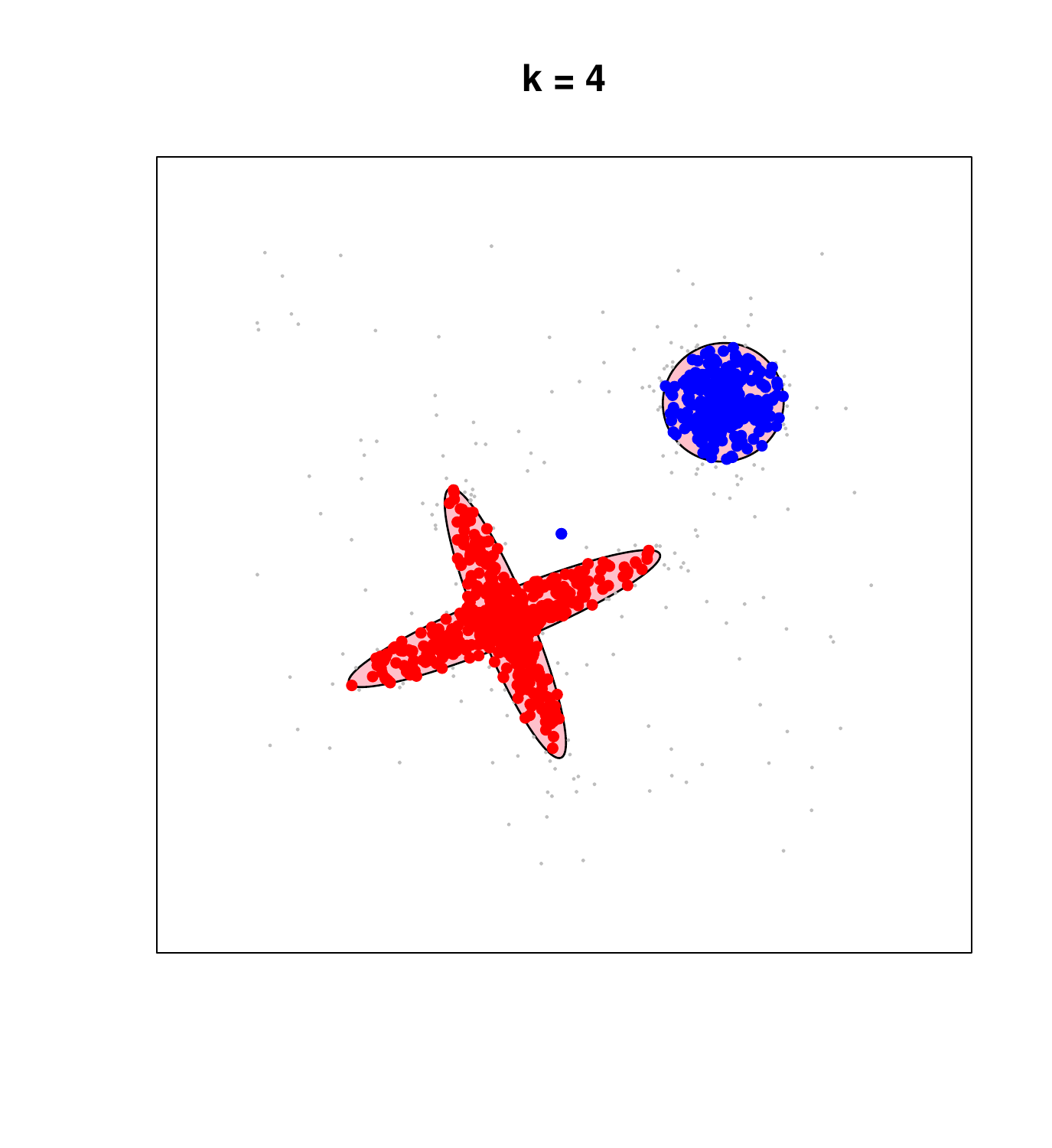}
		\end{center}
		\caption{\em Left: A two-dimensional dataset.
			Middle: $k$ versus volume of clusters.
			Right: Clusters with minimum volume.}
		\label{fig::GMM}
	\end{figure}

	\section{A Max-Mixture Model}
	\label{section::max-mix}
	
	In this section
	we link unions of spheres ---
	and more generally, unions of ellipsoids ---
	to level sets of densities
	by introducing
	a modified mixture model.
	
	Let
	$\theta = \{\pi_j, \mu_j, \Sigma_j\}_{j=1}^k$
	and define the density
	\begin{equation}
	p(y;\theta)=\frac{\max_{1\leq j \leq k} \pi_j \phi(y; \mu_j, \Sigma_j)}{Z_\theta}
	\end{equation}
	where
	$Z_\theta = \int \max_{1\leq j \leq k} \pi_j \phi(y; \mu_j, \Sigma_j)$.
	Let ${\cal F}_k$ be all such densities.
	We call this the {\em max-mixture model.}
	
	The $1-\alpha$ upper level set of $f(y; \theta) \in
	\mathcal{F}_k$ is the union of ellipsoids of the following form:
	\begin{align*}
	\tilde{C}^{(\alpha)}&:= \{y : f(y; \theta) \geq t^{(\alpha)}\}\\ 
	& =
	\bigcup_{j=1}^k \left\{y : (y-\mu_j)^T \Sigma_j^{-1} (y -\mu_j) \leq -
	2\log \det(\Sigma_j) + 2\log \pi_j + M^{(\alpha)}\right\}
	\end{align*}
	where $t^{(\alpha)}, M^{(\alpha)} \in \mathbb{R}$ 
	are constants chosen to make $P( \tilde{C}^{(\alpha)}) \geq 1-\alpha$.
	Hence, for this modified mixture family,
	upper level sets correspond exactly with
	unions of ellipsoids.
	Of course, we strict $\Sigma_j$ to be of the form $\sigma_j^2 I$ then we get
	a union of spheres.
	
	The negative log-likelihood for this model is
	\begin{equation}
	\begin{aligned}
	\ell(\theta):=&\frac{1}{n} \sum_{i=1}^n \min_j 
	\left[ \frac{1}{2} (Y_i-\mu_j)^T \Sigma_j^{-1} (Y_i -\mu_j) + \frac{1}{2}\log \det(\Sigma_j) - \log \pi_j \right] \\
	&+ \log \int \max_{j \in [k]} \pi_j \phi(y; \mu_j, \Sigma_j)(y) \mathrm{d} y.
	\end{aligned}
	\end{equation}
	Minimizing
	$\ell(\theta)$ is difficult
	because of the integral. Solutions of (Generalized) 
	k-means and maximum likelihood of Gaussian mixture problems can be interpreted as two approximate minimizers of $\ell(\theta)$.

	\subsection{(Generalized) k-means : Minimize the first term of $\ell(\theta)$}
	
	We can obtain an approximate minimizer by
	using the first part of $\ell(\theta)$ as the objective function.
	That is, we estimate the parameters by minimizing
	\begin{equation}
	\ell_{kM}(\theta) := 
	\frac{1}{n} \sum_{i=1}^n \min_j \left[ \frac{1}{2} (Y_i-\mu_j)^T \Sigma_j^{-1} (Y_i -\mu_j) + 
	\frac{1}{2}\log \det(\Sigma_j) - \log \pi_j \right].
	\end{equation}
	
	\begin{proposition} \label{prop::kmeans}
		For any $ \theta = \{\pi_j, \mu_j, \Sigma_j\}_{j=1}^k$, we have the following bound :
		\begin{equation}
		\ell_{kM}(\theta)  - \log k \leq \ell (\theta) \leq \ell_{kM}(\theta).
		\end{equation}
		Therefore, for any  $\theta^{\dagger}$ and
		$\theta^*$ 
		minimizing 
		$\ell$ and $\ell_{kM}$ we get
		\begin{align*}
		\left| \ell(\theta^\dagger)  - \ell(\theta^*)  \right| &\leq \log k,\ \ \ {\rm and}\ \ \ 
		\left| \ell_{kM}(\theta^\dagger)  - \ell_{kM}(\theta^*)  \right| \leq \log k .
		\end{align*}
	\end{proposition}
	
	\begin{proof}
		It is enough to show that 
		$$
		-\log k \leq \log \int \max_{j \in [k]} \pi_j \phi(y; \mu_j, \Sigma_j)(y) \mathrm{d} y \leq 0
		$$
		Since each $\phi(y; \mu_j, \Sigma_j)(y)$ is a density function,  the first inequality is followed by
		\begin{align*}
		- \log k &= \log \left(\frac{1}{k} \sum_{j=1}^k \pi_j  \right) =
		\log   \left(\frac{1}{k} \sum_{j=1}^k \pi_j \int \phi(y; \mu_j, \Sigma_j)(y) \mathrm{d} y\right)  \\
		&= \log \int \frac{1}{k} \sum_{j=1}^k \pi_j \phi(y; \mu_j, \Sigma_j)(y) \mathrm{d} y \leq
		\log \int \max_{j \in [k]} \pi_j \phi(y; \mu_j, \Sigma_j)(y) \mathrm{d} y .
		\end{align*}
		The second inequality is follows since
		\begin{align*}
		\log \int \max_{j \in [k]} \pi_j \phi(y; \mu_j, \Sigma_j)(y) \mathrm{d} y 
		&\leq \log \int \sum_{j=1}^k\pi_j \phi(y; \mu_j, \Sigma_j)(y) \mathrm{d} y \\
		& = \log \sum_{j=1}^k\pi_j  \int \phi(y; \mu_j, \Sigma_j)(y) \mathrm{d} y =
		\log \sum_{j=1}^k\pi_j  =0.
		\end{align*}
	\end{proof}
	
	If we restrict to 
	$\Sigma_j \in \{\sigma^2 \mathrm{I} : \sigma \in \mathbb{R}\}$ and set
	$\pi_j = \frac{1}{k}$ for all $j = 1, \dots, k$, then minimizing $\ell$
	is the same as the ordinary $k$-means problem. Generally, Lloyd's algorithm can be directly extended to find a local optimum of the generalized $k$-means problem.
	The details are in Algorithm~\ref{alg::lloyd}.
	
	\begin{algorithm}
		\fbox{\parbox{\textwidth}{
				\begin{center}
					{\sf Generalized Lloyd's Algorithm}
				\end{center}
				\begin{enumerate}
					\item Initialize $\theta_j =  (\pi_j, \mu_j, \Sigma_j)$ for $j = 1, \dots ,k$.
					\item Set 
					\[w_{i,j} = \left\{
					\begin{array}{ll}
					1 & \text{if~~} j = \argmin_l\frac{1}{2} (Y_i-\mu_l)^T \Sigma_l^{-1} (Y_i -\mu_l) + \frac{1}{2}\log \det(\Sigma_l) - \log \pi_l \\
					0 &  \text{otherwise}
					\end{array}
					\right.
					\]
					
					\item Update $\mu_j^+ = \frac{1}{\sum_{i=1}^n w_{i,j}}\sum_{i=1}^n w_{i,j} Y_i$, for all $j = 1, \dots, k$.
					\item Update $\Sigma_j^+ = 
					\frac{1}{\sum_{i=1}^n w_{i,j}}\sum_{i=1}^n w_{i,j} (Y_i- \mu_j^+)(Y_i - \mu_j^+)^T$ , for all $j = 1, \dots, k$.
					\item Update $\pi_j^+ = \frac{1}{n} \sum_{i=1}^n w_{i,j}$, for all $j = 1, \dots, k$.
					\item Repeat step 2-5 until converge.
				\end{enumerate}
		}}
		\caption{\em The generalized version of Lloyd's algorithm for the Max-Mixture Model.}
		\label{alg::lloyd}
	\end{algorithm}

	\subsection{Gaussian Mixture: Minimize another upper and lower bounds of $\ell(\theta)$}
	The negative log-likelihood of Gaussian mixture model can be rewritten as
	\[
	\ell_{GM}(\theta) := 
	-\frac{1}{n} \sum_{i=1}^n \log \left[\sum_{j=1}^k  \exp \left\{-\frac{1}{2} (Y_i-\mu_j)^T \Sigma_j^{-1} (Y_i -\mu_j) - 
	\frac{1}{2}\log \det(\Sigma_j) + \log \pi_j\right\} \right].
	\]
	
	\begin{proposition}
		For any $ \theta = \{\pi_j, \mu_j, \Sigma_j\}_{j=1}^k$, we have the following bound :
		\begin{equation}
		\ell_{GM}(\theta)  - \log k \leq \ell (\theta) \leq \ell_{GM}(\theta) + \log k.
		\end{equation}
		Therefore, for any  $\theta^{\dagger}$ and
		$\theta^*$ 
		minimizing 
		$\ell$ and $\ell_{GM}$ we get
		\begin{align*}
		\left| \ell(\theta^\dagger)  - \ell(\theta^*)  \right| &\leq 2\log k,\ \ \ {\rm and}\ \ \ 
		\left| \ell_{GM}(\theta^\dagger)  - \ell_{GM}(\theta^*)  \right| \leq 2\log k .
		\end{align*}
	\end{proposition}
	
	\begin{proof}
		From Proposition \ref{prop::kmeans}, we know that
		$$
		\ell_{kM}(\theta) - \log k \leq \ell(\theta) \leq \ell_{kM}(\theta).
		$$
		Therefore, it is enough to show that
		\[
		\ell_{GM}(\theta)  \leq \ell_{kM}(\theta) \leq \ell_{GM}(\theta) + \log k
		\]
		
		Note that
		\begin{align*}
		\ell_{kM}(\theta) &= \frac{1}{n} \sum_{i=1}^n \min_j \left[ \frac{1}{2} (Y_i-\mu_j)^T \Sigma_j^{-1} (Y_i -\mu_j) + 
		\frac{1}{2}\log \det(\Sigma_j) - \log \pi_j \right] \\
		& = \frac{1}{n} \sum_{i=1}^n \log \left[\exp \left\{\min_{j \in [k]}  \frac{1}{2} (Y_i-\mu_j)^T \Sigma_j^{-1} (Y_i -\mu_j) + 
		\frac{1}{2}\log \det(\Sigma_j) - \log \pi_j\right\} \right]\\
		& = -\frac{1}{n} \sum_{i=1}^n \log \left[\max_{j \in [k]}  \exp \left\{-\frac{1}{2} (Y_i-\mu_j)^T \Sigma_j^{-1} (Y_i -\mu_j) - 
		\frac{1}{2}\log \det(\Sigma_j) +\log \pi_j\right\} \right]
		\end{align*}
		Hence,  the first inequality comes from
		\begin{align*}
		\ell_{kM}(\theta) &= -\frac{1}{n} \sum_{i=1}^n \log \left[\max_{j \in [k]}  \exp \left\{-\frac{1}{2} (Y_i-\mu_j)^T \Sigma_j^{-1} (Y_i -\mu_j) - 
		\frac{1}{2}\log \det(\Sigma_j) +\log \pi_j\right\} \right]\\
		&\geq -\frac{1}{n} \sum_{i=1}^n \log \left[\sum_{j=1}^k \exp \left\{-\frac{1}{2} (Y_i-\mu_j)^T \Sigma_j^{-1} (Y_i -\mu_j) - 
		\frac{1}{2}\log \det(\Sigma_j) + \log \pi_j\right\} \right] \\
		& = \ell_{GM}(\theta)
		\end{align*}
		And the second inequality is followed by
		\begin{align*}
		\ell_{kM}(\theta) &= -\frac{1}{n} \sum_{i=1}^n \log \left[\max_{j \in [k]}  \exp \left\{-\frac{1}{2} (Y_i-\mu_j)^T \Sigma_j^{-1} (Y_i -\mu_j) - 
		\frac{1}{2}\log \det(\Sigma_j) +\log \pi_j\right\} \right]\\
		&\leq -\frac{1}{n} \sum_{i=1}^n \log \left[\frac{1}{k}\sum_{j=1}^k \exp \left\{-\frac{1}{2} (Y_i-\mu_j)^T \Sigma_j^{-1} (Y_i -\mu_j) - 
		\frac{1}{2}\log \det(\Sigma_j) + \log \pi_j\right\} \right] \\
		& = \ell_{GM}(\theta) + \log k
		\end{align*}
	\end{proof}
	
	The standard EM algorithm can be used to get a local minimum of $\ell_{GM}(\theta)$.
	
	\subsection{Union of ellipsoids clustering - General algorithm.}
	
	To get $1-\alpha$ upper level set of $p(y; \theta)$, we can directly
	use the split conformal prediction method with the residual function
	$-p(y; \theta)$. However, to evaluate $-p(y;\theta)$, we need to
	calculate the normalizing constant$ \int \max_{j \in [k]} \pi_j
	\phi(y; \mu_j, \Sigma_j)(y) $. Since the constant does not depend on
	$y$, we use an equivalent residual function defined by
	\begin{equation}
	\tilde{p}(y ; \theta) = 
	\min_{j \in [k]}  \left[ (y-\mu_j)^T \Sigma_j^{-1} (y -\mu_j) + \log \det(\Sigma_j) - 2\log \pi_j \right].
	\end{equation} 
	This is equivalent to residual functions defined in \eqref{eq::resid} and \eqref{eq::resid_GMM} except how 
	we estimate the parameters. 
	The general split conformal steps are in Algorithm~\ref{alg::ellipse}.
	
	\begin{algorithm}
		\fbox{\parbox{\textwidth}{
				\begin{center}
					{\sf Unions of Ellipses}
				\end{center}
				
				\begin{enumerate}
					\item Split the data into two halves
					${\cal Y}_1$ and 
					${\cal Y}_2$.
					\item Estimate $\hat{\theta}$ from ${\cal Y}_1$ by solving generalized k-means or maximum likelihood problem of GMM.  
					\item Compute the residuals
					$$
					R_i = \tilde{f}(Y_i ; \hat{\theta}) = 
					\min_{j \in [k]}  
					\left[ (Y_i-\hat{\mu}_j)^T \hat{\Sigma}_j^{-1} (Y_i -\hat{\mu}_j) + 
					\log \det(\hat{\Sigma}_j) - 2\log \hat{\pi}_j \right]
					$$
					where
					$Y_i \in {\cal Y}_2$.
					\item Let $M^{(\alpha)}$ be the $1-\alpha$ quantile of the residuals.
					\item Let 
					\begin{align}
					\tilde{C}^{(\alpha)}&= \{y : \tilde{f}(y; \hat{\theta}) \leq M^{(\alpha)}\}\\
					& = 
					\bigcup_{j=1}^k \left\{y : (y-\mu_j)^T \Sigma_j^{-1} (y -\mu_j) \leq - 2\log \det(\Sigma_j) + 
					2\log \pi_j + M^{(\alpha)}\right\}.
					\end{align}
				\end{enumerate}
		}}
		\caption{Obtaining a union of ellipses from the generalized Lloyd's algorithm.}
		\label{alg::ellipse}
	\end{algorithm}
	
	We now see that
	$\tilde{C}^{(\alpha)}$ is a union of ellipsoids with $1-\alpha$ coverage, that is, 
	$$
	\mathbb{P}(Y_{n+1} \in \tilde{C}^{(\alpha)}) = P(\tilde{C}^{(\alpha)}) \geq 1-\alpha
	$$
	Unlike common mixture model based clustering, our method defines
	clusters as connected components of a level set which make it possible
	to apply our method to general shaped clusters.

	Figure \ref{fig::max-mixture} shows an example.
	Panel (a) shows the data with 4 crescent-shaped clusters.  Panel (b -
	d) show the resulting clusters in which parameter estimations are
	based on $k$-means, generalized $k$-means, and a Gaussian mixture
	model, respectively. Each union of balls/ellipsoids is a $1-\alpha$
	predictive region of the underlying distribution with $\alpha =
	0.1$. Although some of the methods yield small erroneous clusters, all
	method recover the crescent-shaped areas reasonably.  
	
	\begin{figure}
		\centering
		\centering
		\begin{subfigure}[b]{0.475\textwidth}
			\centering
			\includegraphics[width=\textwidth]{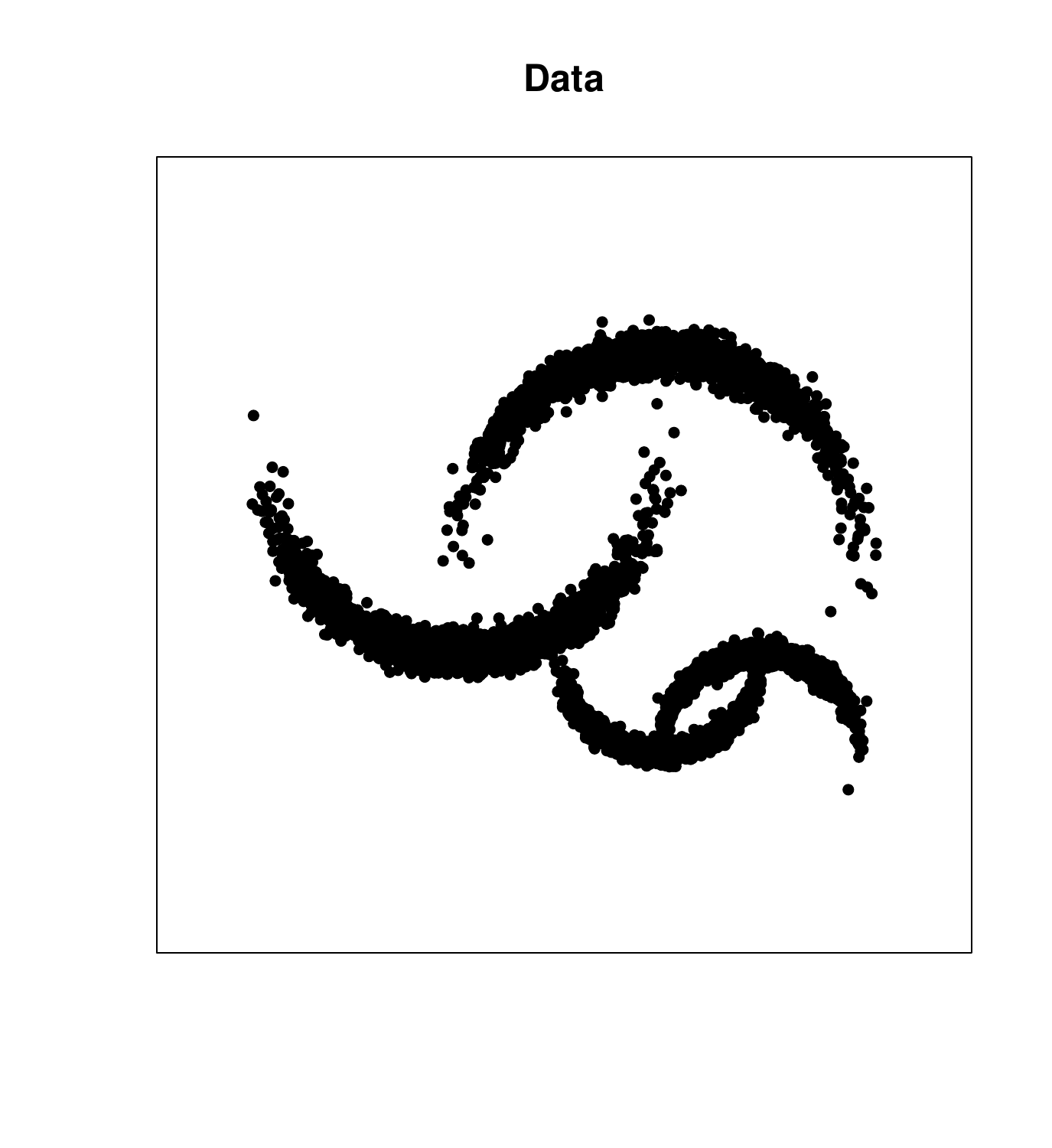}
		\end{subfigure}
		\hfill
		\begin{subfigure}[b]{0.475\textwidth}  
			\centering 
			\includegraphics[width=\textwidth]{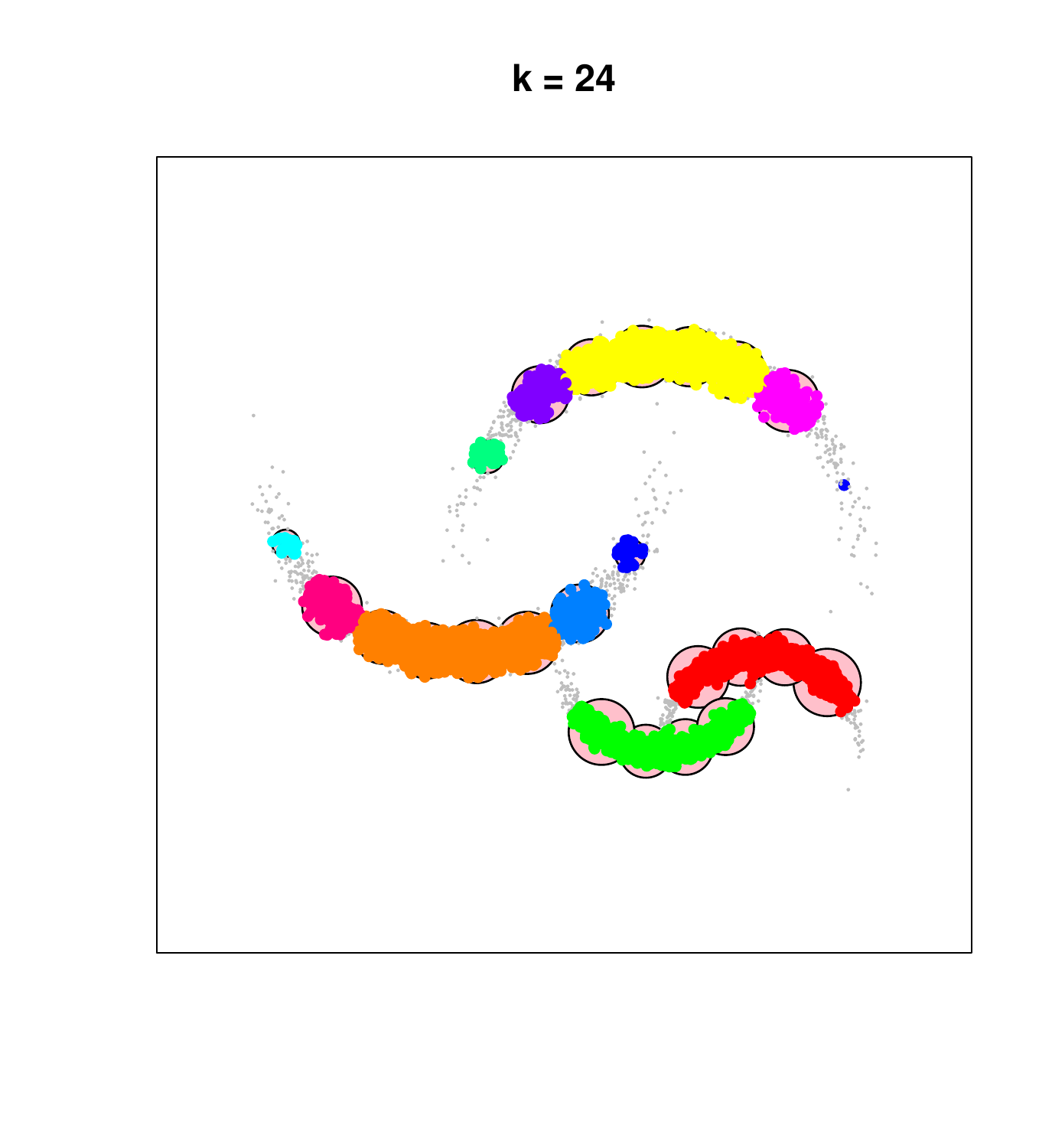}
		\end{subfigure}
		\vskip\baselineskip
		\begin{subfigure}[b]{0.475\textwidth}   
			\centering 
			\includegraphics[width=\textwidth]{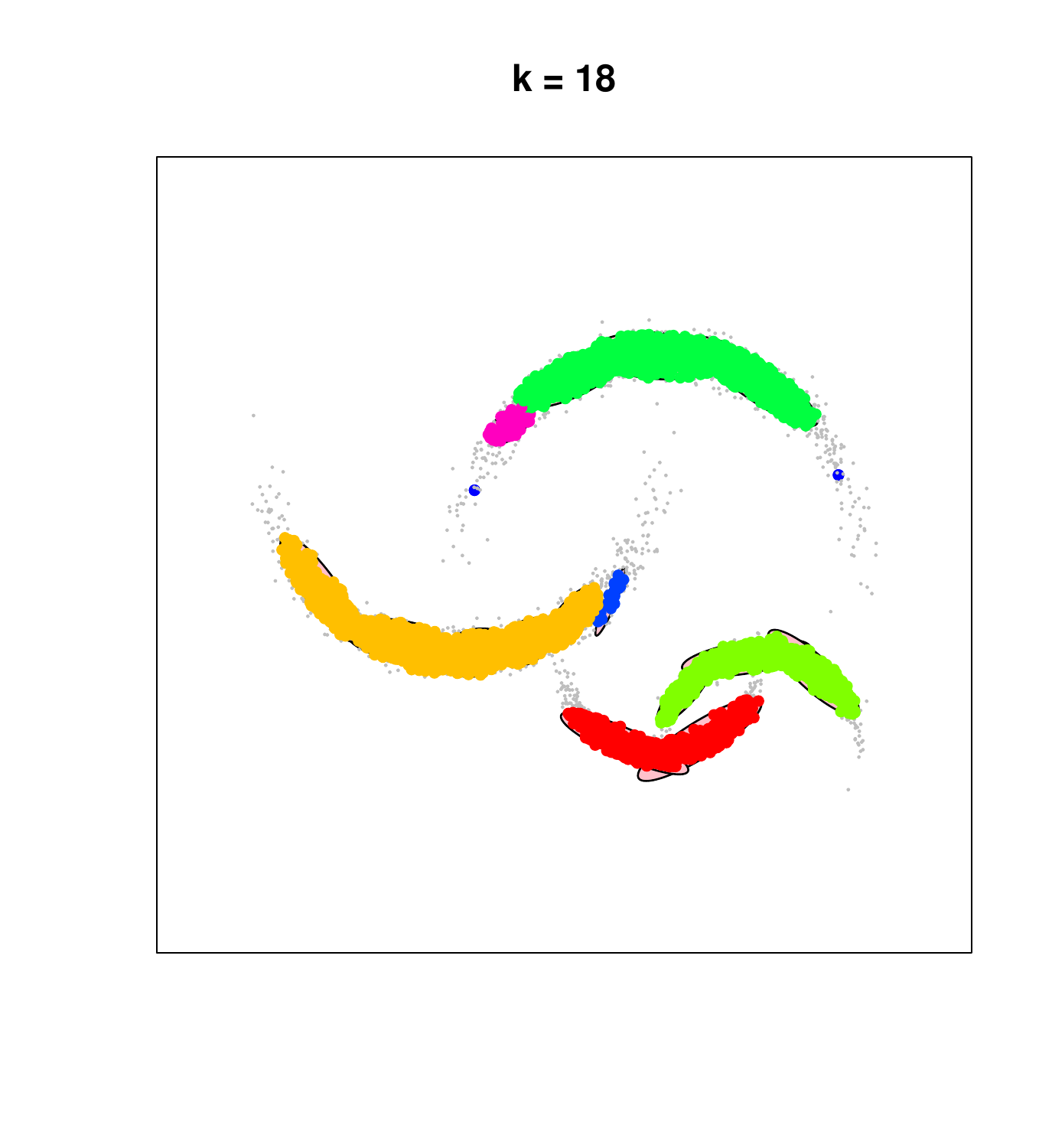}
		\end{subfigure}
		\quad
		\begin{subfigure}[b]{0.475\textwidth}   
			\centering 
			\includegraphics[width=\textwidth]{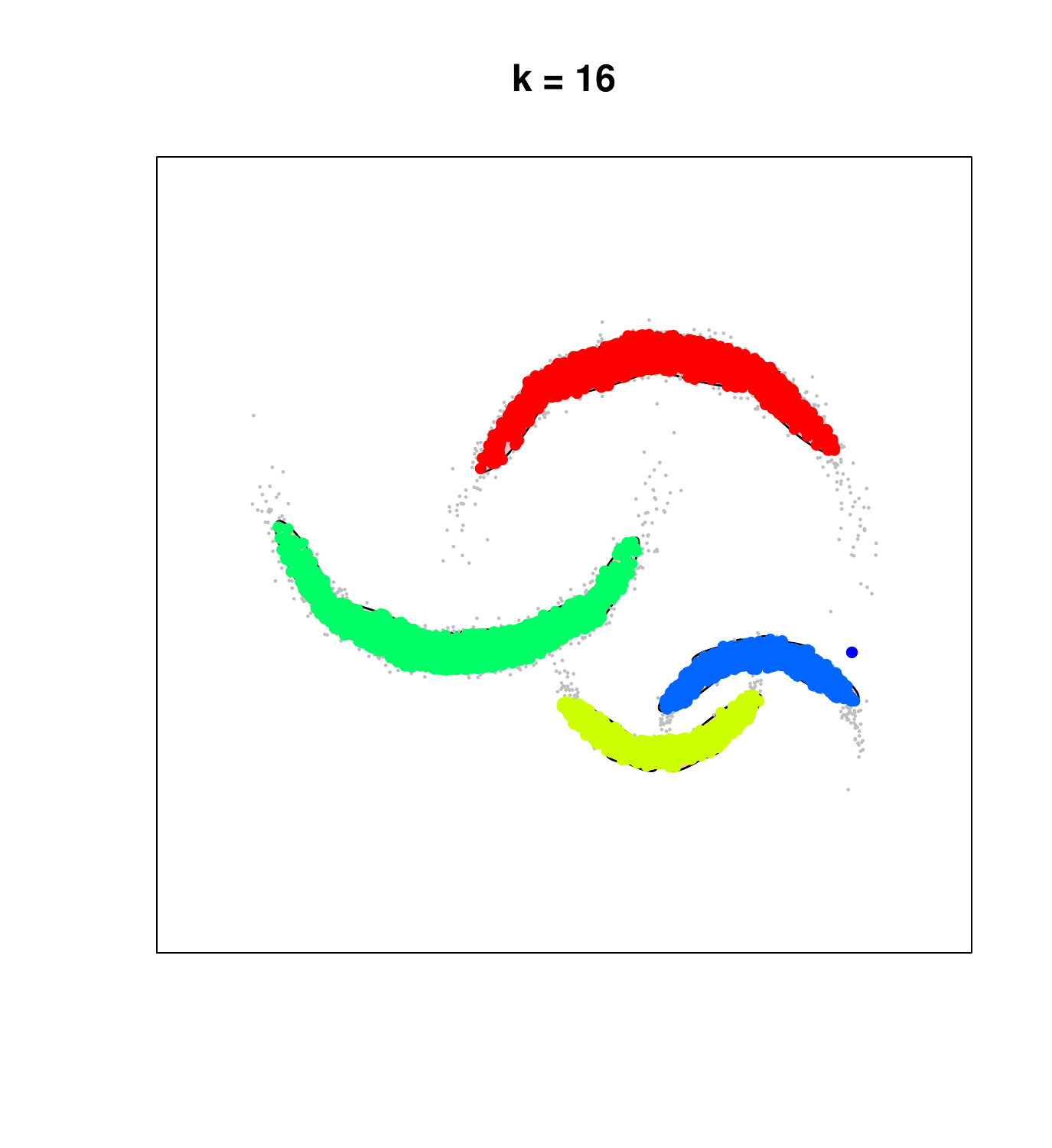}
		\end{subfigure}
		\caption{\em (a) A two-dimensional dataset.
			(b) k-means.
			(c) Generalized k-means.
			(d)  GMM.}
		\label{fig::max-mixture}
	\end{figure}

	\section{Converting Level Sets to Unions of Spheres}
	\label{section::level-sets}
	
	Another approach to clustering
	is to use the upper level sets
	of density estimates.
	These can easily be converted into unions of spheres.
	We focus on $k$-nearest neighbor density estimation.

	\subsection{Density Functions}
	
	Let
	$L=\{y:\ \hat p(y) >t\}$
	be the upper level set of a nonparametric density estimator $\hat p$.
	The set $L$
	is often used for clustering;
	see \cite{cuevas2001cluster}, for example.
	In general, $L$
	is a very complex set and is hard to compute.
	In general, there is no simple way to represent this set.
	Some authors have suggested approximating
	$L$ with a union of spheres.
	In this section, we show that the conformalization approach
	automatically turns $L$ into a union of spheres and that, as with the other approaches,
	minimizing volume allows us to choose all tuning parameters.
	The details are in Algorithm~\ref{alg::level}.
	
	\begin{algorithm}
		\fbox{\parbox{\textwidth}{
				\begin{center}
					{\sf Density Level Set Conversion}
				\end{center}
				\begin{enumerate}
					\item Split the data into two halves
					${\cal Y}_1$ and ${\cal Y}_2$.
					\item Estimate a density $\hat{p}$ from ${\cal Y}_1$.
					\item For a given $t > 0$, let ${\cal Y}_1^{(t)} := \{Y_i \in {\cal Y}_1 : \hat{p}(Y_i) \geq t\}$.
					\item Compute the residuals $R_i := R(Y_i)$ for $Y_i \in {\cal Y}_2$, where
					$R(y)= \min_{Y_j \in {\cal Y}_1^{(t)}}  \|y - Y_j\|$.
					\item Let $M^{(\alpha)}$ be the $1-\alpha$ quantile of the residuals.
					\item Let 
					$\hat{C}^{(\alpha)}= \{y : R(y) \leq M^{(\alpha)}\}=
					\bigcup_{Y_j \in {\cal Y}_1^{(t)}} \left\{y : \|y - Y_j\| \leq M^{(\alpha)}.\right\}$
				\end{enumerate}
		}}
		\caption{Algorithm to convert density level set into a union of spheres.}
		\label{alg::level}
	\end{algorithm}
	
	Recall that the $k$-nn density estimate is
	$\hat p(x) = k/(n d_k(x))$
	where $d_k(x)$ is the distance from $x$ to its
	$k^{\rm th}$ nearest-neighbor.
	The resulting set $\hat{C}^{(\alpha)}$ is a union of balls with $1-\alpha$ coverage, that is, 
	$$
	P^{n+1}(Y_{n+1} \in \hat{C}^{(\alpha)}) = P(\hat{C}^{(\alpha)}) \geq 1-\alpha.
	$$
	We can minimize the Lebesgue measure to choose both tuning parameters
	$t$ and $k$.

	\begin{example}
		Figure \ref{fig::KNN} illustrates the results of Algorithm~\ref{alg::level}.
		We use the
		same data with the crescent-shaped clusters as in the
		previous example. Panel (a) shows how the volume
		of the clusters changes as the parameter $k$ varying from
		$2^1$ to $2^8$ with the level $t$ fixed to the
		$1-\alpha$ quantile of density values from the first
		half data $\mathcal{Y}_1$. The shaded lines
		corresponds to volume curves with bootstrapped data.
		Panel (b) shows clusters with $k = 32$ which is
		selected by the testing based method. Panel (c)
		illustrates how the volume of the clusters changes as the
		level parameter $t$ varying from $1-\alpha / 2$ to $1-
		2\alpha$ quantiles of the density values of the first half
		data $\mathcal{Y}_1$ with fixed $k = 32$. The shaded
		lines corresponds to volume curves with bootstrapped
		data. Panel (d) shows clusters with $k = 32$, $t$
		equal to the $1-0.112$ quantile of density values on the
		first half data $\mathcal{Y}_1$.
	\end{example}

	\begin{figure}
		\centering
		\centering
		\begin{subfigure}[b]{0.475\textwidth}
			\centering
			\includegraphics[width=\textwidth]{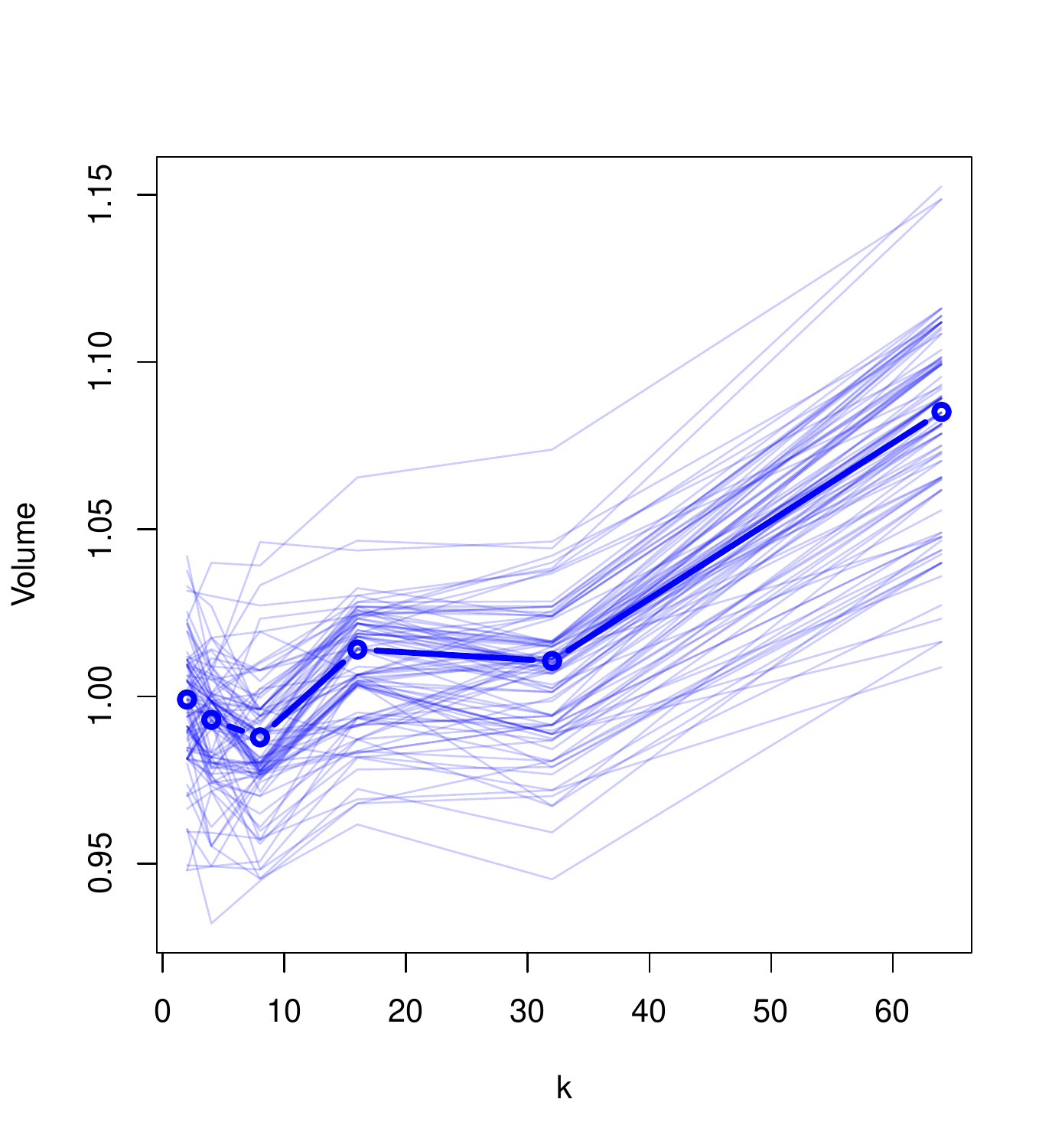}
		\end{subfigure}
		\hfill
		\begin{subfigure}[b]{0.475\textwidth}  
			\centering 
			\includegraphics[width=\textwidth]{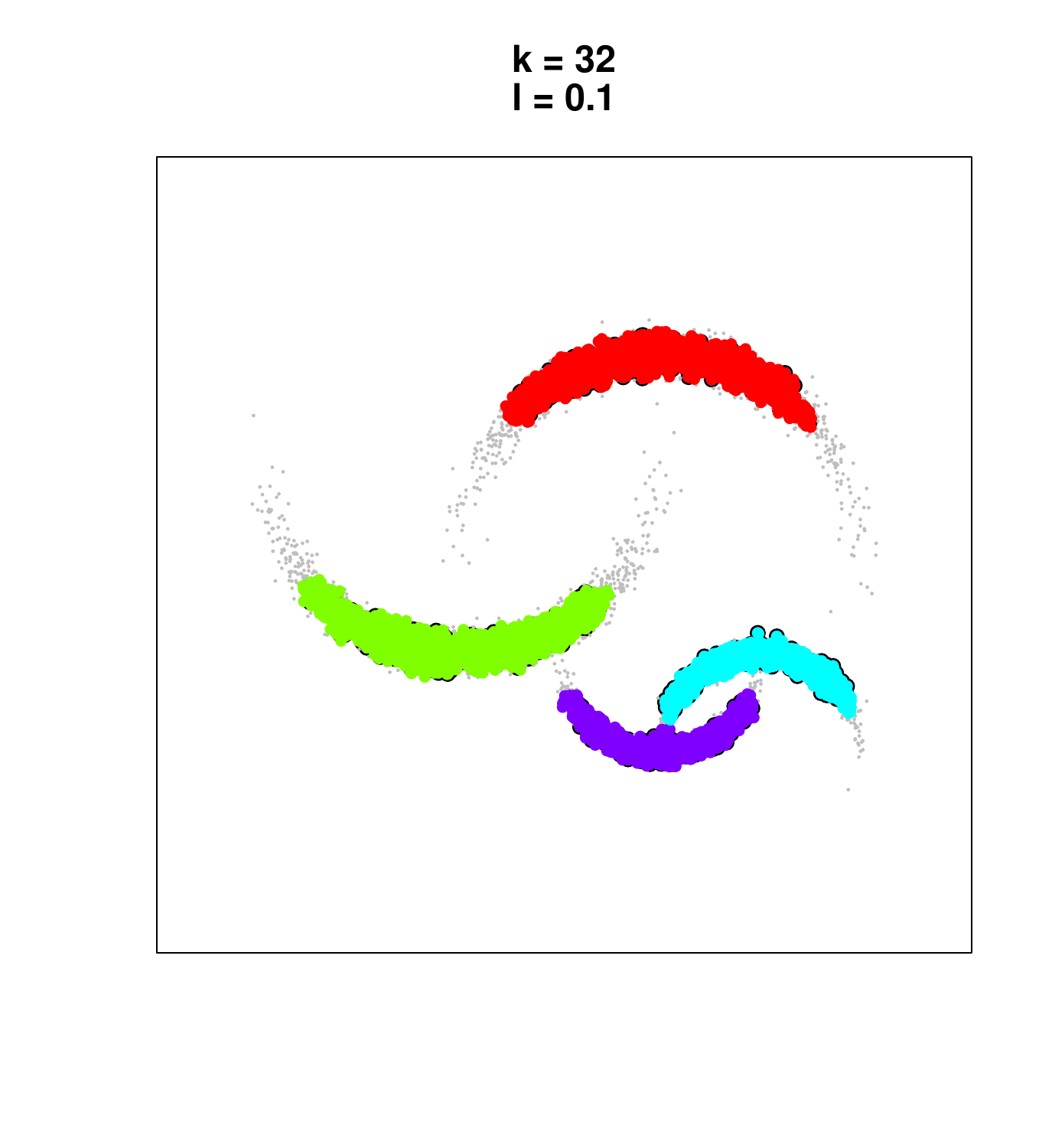}
		\end{subfigure}
		\vskip\baselineskip
		\begin{subfigure}[b]{0.475\textwidth}   
			\centering 
			\includegraphics[width=\textwidth]{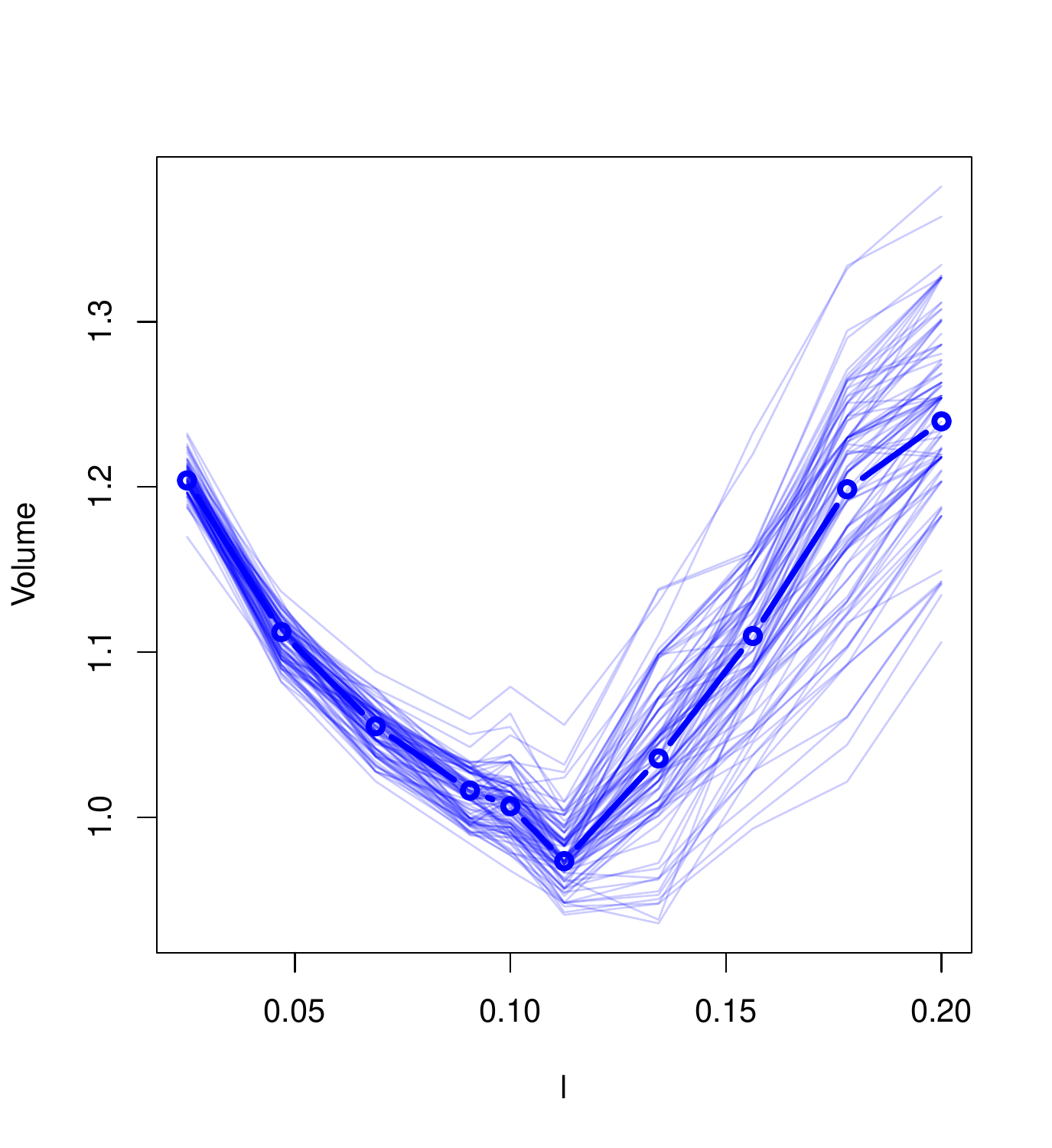}
		\end{subfigure}
		\quad
		\begin{subfigure}[b]{0.475\textwidth}   
			\centering 
			\includegraphics[width=\textwidth]{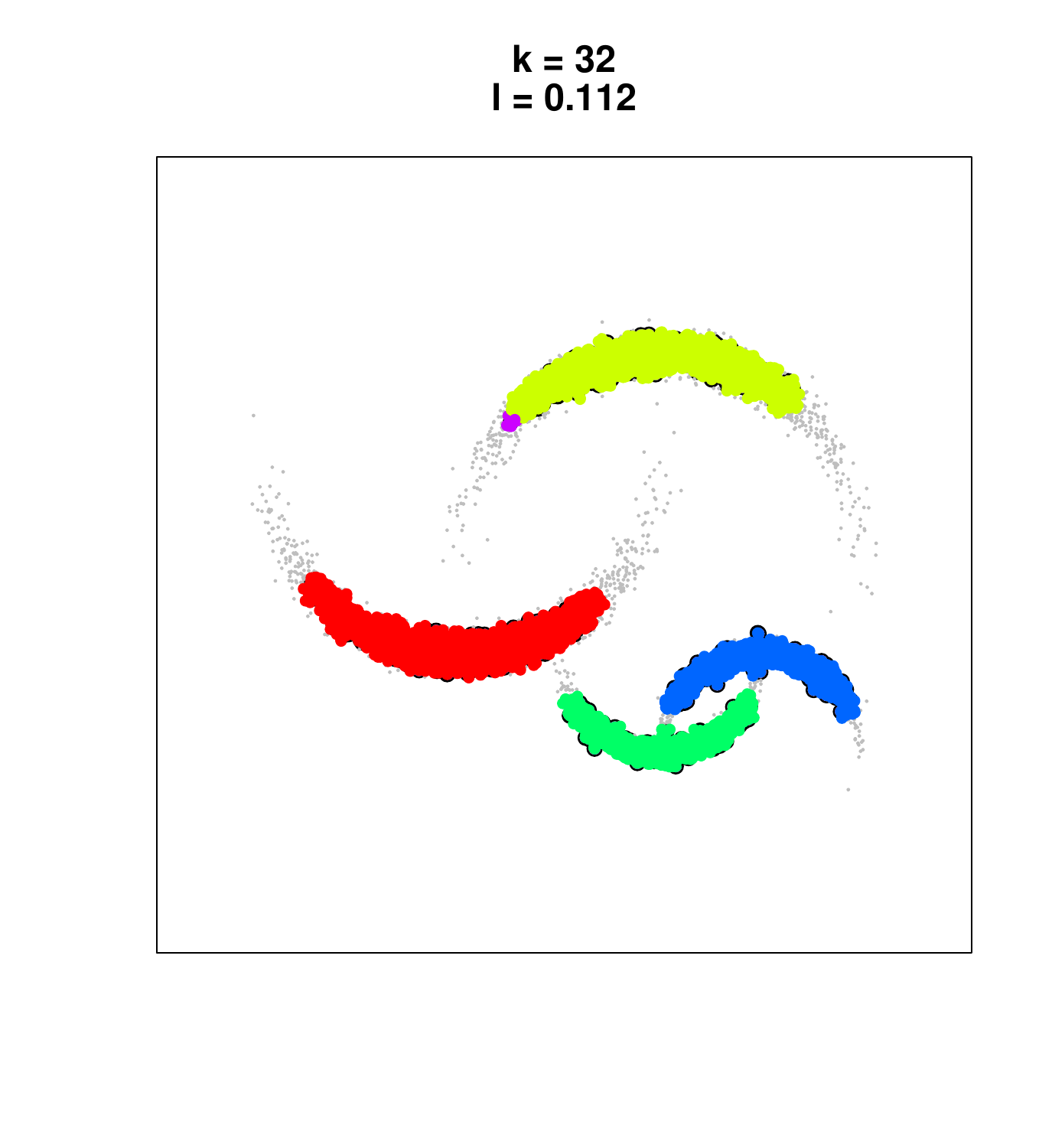}
		\end{subfigure}
		\caption{\em (a) $k$ versus volume curve with fixed level $t >0$.
			(b) k-NN level set clusters with $\hat{k}$.
			(c)  $t$ versus volume curve with fixed parameter $k = \hat{k}$.
			(d) k-NN level set clusters with $\hat{k}$ and $\hat{t}.$}
		\label{fig::KNN}
	\end{figure}

	{\bf Remark:}
	{\em As with $k$-means we may want to use spheres of varying size.  We can
		do this by introducing an adaptive residual function such as $R(y) =
		\min_{Y_j \in {\cal Y}_1^{(t)}} \|y - Y_j\|\sqrt{\hat{p}(Y_j)}$ for
		example.  Then,
		$$
		\hat{C}^{(\alpha)}= \{y : R(y) \leq M^{(\alpha)}\}=
		\bigcup_{Y_j \in {\cal Y}_1^{(t)}} \left\{y : \|y - Y_j\| \leq \frac{ M^{(\alpha)} }{ \sqrt{\hat{p}(Y_j)} } \right\}.
		$$
	}

	\section{Conclusion and Future Work}
	
	In this paper
	we showed that we can construct distribution-free
	prediction sets from clustering.
	These prediction sets can be regarded as an improved clustering.
	We can choose the tuning parameters
	in the clustering to minimize the volume of the prediction set.
	
	There are some natural extensions
	of this approach.
	First, it would be interesting to create a streaming version of the method.
	In fact, the conformal prediction was originally presented as a sequential method.
	It may be possible, therefore, to incorporate streaming clustering as well.
	Second, there have been some papers
	on clustering in the regression framework
	\cite{loubes2017prediction,chen2018modal}.
	In these models,
	we let the clustering structure of $Y$ vary with $x$.
	We could adapt our methods to produce
	cluster-structured prediction sets.
	Such sets can be much smaller than
	standard regression prediction sets
	if there is a clustering structure in the data.

	\bibliographystyle{plainnat}
	\bibliography{paper}

\end{document}